\documentclass[11pt,a4paper,leqno]{amsart}

\usepackage[utf8]{inputenc}
\usepackage[english]{babel}
\usepackage{amsmath,amsfonts,amssymb,amsthm}
\usepackage{xcolor}
\usepackage{tabularx}
\usepackage{tikz}
\usepackage{floatrow}
\usepackage{multicol} %plusieurs colonnes itemize
\usepackage{pgfplots}
\usepackage{enumitem}

\theoremstyle{plain}
\newtheorem{theorem}{Theorem}[section]
\newtheorem*{theorem*}{Theorem}
\newtheorem{corollaire}[theorem]{Corollary}
\newtheorem{lemme}[theorem]{Lemma}

\theoremstyle{definition}

\newtheorem{rmq}[theorem]{Remark}

\makeatletter

\@addtoreset{equation}{section}
\makeatother

\makeatletter
\renewcommand{\thefigure}{\thesection.\@arabic\c@figure}
\@addtoreset{figure}{section}
\makeatother

\newcommand{\joinf}{\lfloor j_0\rfloor}
\newcommand{\bin}[1]{\mathcal{B}_{#1}}
\newcommand{\binc}[1]{\mathcal{B}^*_{#1}}

\begin{document}

\author{Karl Grosse-Erdmann}
\dedicatory{Université de Mons}
\address{Institut Complexys, Département de Mathématique,
Université de Mons, 20~Place du Parc, 7000 Mons, Belgium.}
\email{kg.grosse-erdmann@umons.ac.be}

\author{Fabien Heuwelyckx}
\dedicatory{Université de Mons}
\address{Institut Complexys, Département de Mathématique,
Université de Mons, 20~Place du Parc, 7000 Mons, Belgium.}
\email{fabien.heuwelyckx@umons.ac.be}

\title[Lookback options and binomial approximation]{The pricing of
  lookback options and binomial approximation}


\thanks{The authors have been supported by a grant from the National
  Bank of Belgium (BNB)}


\begin{abstract}
  Refining a discrete model of Cheuk and Vorst we obtain a closed
  formula for the price of a European lookback option at any time
  between emission and maturity. We derive an asymptotic expansion of
  the price as the number of periods tends to infinity, thereby
  solving a problem posed by Lin and Palmer. We prove, in particular,
  that the price in the discrete model tends to the price in the
  continuous Black-Scholes model. Our results are based on an
  asymptotic expansion of the binomial cumulative distribution
  function that improves several recent results in the literature.
\end{abstract}

\maketitle

\section{Introduction}

Cheuk and Vorst~\cite{cheuk_vorst} have proposed a discrete model for
the pricing of European lookback options with floating strike; they
suppose implicitly that they evaluate the price at emission. In this
paper we will refine their model in order to price the option at any
given time after emission, and we derive an asymptotic expansion for
the price as the number of time intervals tends to infinity. This
solves completely a problem posed by Lin and Palmer~\cite{lin_palmer};
in the special case of the price at emission the problem has already
been treated by the second author~\cite{heuw1}.

Lookback options give the holder the right to buy (for the call),
respectively to sell (for the put), the underlying asset at maturity
for its lowest, respectively highest, price during its lifetime. The
payoff functions at maturity are therefore given by
\begin{equation*}
S_T-\min_{t\leq T}S_t\quad
\text{and}\quad
\max_{t\leq T}S_t-S_T,
\end{equation*}
respectively.

In the traditional continuous model (i.e., when the underlying asset
price follows a Wiener process with drift, as proposed by Black and
Scholes~\cite{BS} and Merton~\cite{merton}), Goldman, Sosin and
Gatto~\cite{gatto_goldman_sosin} derived a formula for the price under
the assumption that the spot rate $r$ is non-zero. Babbs~\cite{babbs}
obtained the price for $r=0$ by passing to the limit.

A discrete model for the price of a lookback option was proposed by
Hull and White~\cite{hull_white}, see also Hull \cite{hull}, who based
it on the familiar binomial model of Cox, Ross and
Rubinstein~(CRR)~\cite{CRR}, see Figure~\ref{fig-crr}.  However, since
lookback options are path-dependent, Hull and White had to subdivide
each node into different states. This problem was overcome by Cheuk
and Vorst~\cite{cheuk_vorst} who proposed an equivalent tree (CV) in
which every node corresponds to a single state, see
Figure~\ref{fig-cv} (for a call).

{\setlength{\captionmargin}{32.5pt}
\begin{figure}[!ht]
\begin{floatrow}
\ffigbox{\begin{tikzpicture}[scale=1.5]
\draw (0,0) -- (1,0.25) -- (2,0.5) -- (3,0.75);
\draw (0,0) -- (1,-0.25) -- (2,-0.5) -- (3,-0.75);
\draw (1,0.25) -- (2,0) -- (3,-0.25);
\draw (2,0.5) -- (3,0.25);
\draw (1,-0.25) -- (2,0) -- (3,0.25);
\draw (2,-0.5) -- (3,-0.25);
\filldraw (0,0) circle(1.25pt);
\filldraw (1,-0.25) circle(1.25pt);
\filldraw (1,0.25) circle(1.25pt);
\filldraw (2,-0.5) circle(1.25pt);
\filldraw (2,0) circle(1.25pt);
\filldraw (2,0.5) circle(1.25pt);
\filldraw (3,-0.75) circle(1.25pt);
\filldraw (3,-0.25) circle(1.25pt);
\filldraw (3,0.25) circle(1.25pt);
\filldraw (3,0.75) circle(1.25pt);
\end{tikzpicture}}{
\caption{CRR tree with $n=3$}\label{fig-crr}}
\ffigbox{\begin{tikzpicture}[scale=1.5]
\draw (0,0) -- (1,0) -- (2,0) -- (3,0);
\draw (0,0) -- (1,0.5) -- (2,1) -- (3,1.5);
\draw (1,0) -- (2,0.5) -- (3,1);
\draw (1,0.5) -- (2,0) -- (3,0.5);
\draw (2,0.5) -- (3,0);
\draw (2,1) -- (3,0.5);
\filldraw (0,0) circle(1.25pt);
\filldraw (1,0) circle(1.25pt);
\filldraw (1,0.5) circle(1.25pt);
\filldraw (2,0) circle(1.25pt);
\filldraw (2,0.5) circle(1.25pt);
\filldraw (2,1) circle(1.25pt);
\filldraw (3,0) circle(1.25pt);
\filldraw (3,0.5) circle(1.25pt);
\filldraw (3,1) circle(1.25pt);
\filldraw (3,1.5) circle(1.25pt);
\end{tikzpicture}}{
\caption{CV tree for the call with $n=3$}\label{fig-cv}}
\end{floatrow}
\end{figure}}

Neither Hull and White, nor Cheuk and Vorst obtained a closed formula
for the prices in their model. Such a formula was first obtained by
F\"ollmer and Schied \cite{follmer_schied} and, by a different method,
by the second author~\cite{heuw1} (see also the discussion in
\cite[Appendix A]{heuw1}). However, all these papers only evaluate the
price of the option at emission.

The aim of this paper is two-fold. In Section 2 we will derive a
closed formula, in the discrete model, for the price at any time after
emission of a European lookback option with floating strike. In
Section 4 we show that the price in the discrete model converges to
the price in the continuous model, and we derive an asymptotic
expansion. This answers a problem posed by Lin and
Palmer~\cite{lin_palmer} and generalizes the earlier results in
\cite{heuw1}. In order to derive these asymptotics we need a
refinement of the known asymptotic expansions of the binomial
cumulative distribution function, which will be achieved in
Section~3. The final section is devoted to numerical examples.

\section{Lookback options with floating strike}

In the discrete model of Cheuk and Vorst, the traditional CRR tree,
see Figure~\ref{fig-crr}, is still used for modelling the evaluation
of the underlying.

We first introduce the usual notations for the various parameters:
$T$~is the time from emission to maturity, $t$ with $0\leq t<T$ is the
present time, $\tau=T-t$~is the remaining time until maturity, and
$S_t$~is the value of the underlying at time~$t$. Moreover, $r\geq
0$~is the spot rate and $\sigma>0$~is the volatility of the underlying
asset.

Now, the binomial tree for the underlying is only built for the time
interval $[t,T]$ (the prices before $t$ being known). This interval is
divided equally into $n$ subintervals. At each node, the price may
either increase by a factor $u$ or it may decrease by a factor $d$,
where
\begin{equation*}
u=e^{\sigma\sqrt{\tau/n}}\quad
\text{and}\quad
d=u^{-1}=e^{-\sigma\sqrt{\tau/n}}.
\end{equation*}
The probability of an increase is given by
\begin{equation}\label{pe}
p=\frac{e^{r\tau/n}-d}{u-d};
\end{equation}
we take $n$ sufficiently large so that $0< p< 1$.

Let us first consider $t=0$, and let us concentrate on call
options. For pricing a lookback option, Cheuk and Vorst have
introduced a second tree (CV), see Figure \ref{fig-cv}. In the CV
tree, the level $j$ at time $t_m$ ($0\leq m\leq n$) denotes the
difference in powers of~$u$ between the value $S_{t_m}$ at time $t_m$
and the lowest value of the underlying since emission. In other words,
$j$ is the non-negative integer such that
\[
S_{t_m}=\big(\min_{t^*\leq t_m}S_{t^*}\big)u^j.
\]
As shown in~\cite{cheuk_vorst} and made explicit in \cite[Theorem
2.1]{heuw1}, the price of the option at emission is then given by
\[
C^{fl}_n(0) = S_0 \sum_{j=0}^n (1-u^{-j})\sum_{k=0}^{n} \Lambda_{j,k,n} q^k(1-q)^{n-k},
\]
where
\begin{equation}\label{qu}
q=pu e^{-r\tau/n}=\frac{u-e^{-r\tau/n}}{u-d} 
\end{equation}
(which lies in $(0,1)$ if $p$ does) and $\Lambda_{j,k,n}$ is the
number of paths in the CV tree from the initial node $(0,0)$ to level
$j$ at maturity that have exactly $k$~upward jumps. It was shown in
\cite[Theorem 2.1]{heuw1} that
\begin{equation*}
\Lambda_{j,k,n}=
\begin{cases} 
\binom{n}{k-j}-\binom{n}{k-j-1}
	& \text{if $j\leq k\leq\lfloor\frac{n+j}{2}\rfloor$},\\
0   & \text{else.}
\end{cases}
\end{equation*}

In this paper we are interested in finding the price of a lookback
(call) option at any given time $t$ $(0\leq t<T)$. In that case, the
market prices $S_{t^*}$, $0\leq t^*\leq t$, of the underlying between
emission and time $t$ are known. Of course, the price at time~$t$ is
not necessarily the minimal value; there is in fact an initial level
$j_0$ that is defined by
\begin{equation*}
S_{t}=\big(\min_{t^*\leq t}S_{t^*}\big)u^{j_0}.
\end{equation*}
In the sequel we will write for brevity
\[
M_{t}= \min_{t^*\leq t}S_{t^*}.
\]
Then we have 
\begin{equation}\label{j0}
j_0=\frac{\log (S_{t}/M_{t})}{\sigma\sqrt{\tau/n}}.
\end{equation}
Clearly, $j_0\geq 0$ is not necessarily an integer, and the initial
node is located at position~$(0,j_0)$. This leads to a modified CV
tree (see Figure~\ref{fig-cheuk}, where $m$ corresponds to time
$t_m$). Note that after sufficiently many downward jumps one reaches a
new minimum and thus level 0. From there on, all levels are integers.

It can be shown as in \cite{cheuk_vorst} that
\begin{equation}\label{price}
C^{fl}_n(t) = S_{t} \sum_{j\in J} (1-u^{-j})\sum_{k=0}^n \Lambda^{j_0}_{j,k,n} q^k(1-q)^{n-k},
\end{equation}
where $J$ is the set of possible levels at maturity, $q$ is given by
\eqref{qu} and, as before, $\Lambda^{j_0}_{j,k,n}$ is the number of
paths from the initial node $(0,j_0)$ to level $j$ at maturity that
have exactly $k$ upward jumps. It remains to evaluate these numbers.

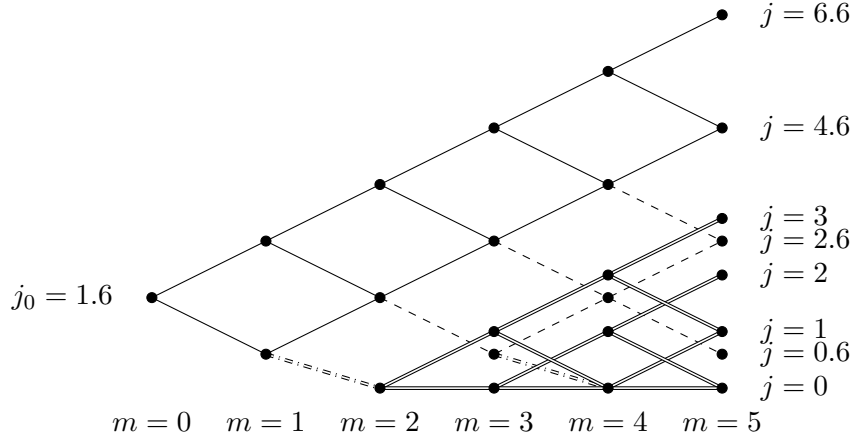
\begin{figure}[ht!]
\begin{center}
\begin{tikzpicture}[scale=0.75]
\draw (0,1.6) -- (2,2.6) -- (4,3.6) -- (6,4.6) -- (8,5.6) -- (10,6.6);
\draw (0,1.6) -- (2,0.6) -- (4,1.6) -- (6,2.6) -- (8,3.6) -- (10,4.6);
\draw (2,2.6) -- (4,1.6);
\draw (4,3.6) -- (6,2.6);
\draw (6,4.6) -- (8,3.6);
\draw (8,5.6) -- (10,4.6);
\draw[dashed] (4,1.6) -- (6,0.6) -- (8,1.6) -- (10,2.6);
\draw[dashed] (6,2.6) -- (8,1.6) -- (10,0.6);
\draw[dashed] (8,3.6) -- (10,2.6);
\draw[double,dashdotted] (2,0.6) -- (4,0);
\draw[double,dashdotted] (6,0.6) -- (8,0);
\draw[double] (4,0) -- (6,1) -- (8,2) -- (10,3);
\draw[double] (4,0) -- (6,0) -- (8,1) -- (10,2);
\draw[double] (6,1) -- (8,0) -- (10,1);
\draw[double] (8,2) -- (10,1);
\draw[double] (8,1) -- (10,0);
\draw[double] (6,0) -- (8,0) -- (10,0);
\filldraw (0,1.6) circle(2.5pt);
\filldraw (2,2.6) circle(2.5pt);
\filldraw (2,0.6) circle(2.5pt);
\filldraw (4,3.6) circle(2.5pt);
\filldraw (4,1.6) circle(2.5pt);
\filldraw (6,4.6) circle(2.5pt);
\filldraw (6,2.6) circle(2.5pt);
\filldraw (6,0.6) circle(2.5pt);
\filldraw (8,5.6) circle(2.5pt);
\filldraw (8,3.6) circle(2.5pt);
\filldraw (8,1.6) circle(2.5pt);
\filldraw (10,6.6) circle(2.5pt);
\filldraw (10,4.6) circle(2.5pt);
\filldraw (10,2.6) circle(2.5pt);
\filldraw (10,0.6) circle(2.5pt);
\filldraw (4,0) circle(2.5pt);
\filldraw (6,0) circle(2.5pt);
\filldraw (6,1) circle(2.5pt);
\filldraw (8,0) circle(2.5pt);
\filldraw (8,1) circle(2.5pt);
\filldraw (8,2) circle(2.5pt);
\filldraw (10,0) circle(2.5pt);
\filldraw (10,1) circle(2.5pt);
\filldraw (10,2) circle(2.5pt);
\filldraw (10,3) circle(2.5pt);
\node[left] at (-0.5,1.6){$j_0=1.6$};
\node[right] at (10.5,6.6){$j=6.6$};
\node[right] at (10.5,4.6){$j=4.6$};
\node[right] at (10.5,3){$j=3$};
\node[right] at (10.5,2.6){$j=2.6$};
\node[right] at (10.5,2){$j=2$};
\node[right] at (10.5,1){$j=1$};
\node[right] at (10.5,0.6){$j=0.6$};
\node[right] at (10.5,0){$j=0$};
\node[below] at (0,-0.25){$m=0$};
\node[below] at (2,-0.25){$m=1$};
\node[below] at (4,-0.25){$m=2$};
\node[below] at (6,-0.25){$m=3$};
\node[below] at (8,-0.25){$m=4$};
\node[below] at (10,-0.25){$m=5$};
\end{tikzpicture}
\caption{CV tree for the call with $n=5$ and $j_0=1.6$}
\label{fig-cheuk}
\end{center}
\end{figure}

Put abstractly, we have the following problem. Let $j_0\geq 0$ be a
positive real number and $n\in\mathbb{N}$. We create a graph with an
initial node at $(0,j_0)$, and from each node $(m,j_m)$ at period $m$,
$0\leq m <n$, we create two nodes at period $m+1$ given by
\[
(m+1,j_m+1)\;\;\textit{up}
\]
and
\[
(m+1, \max( j_m-1,0))\;\;\textit{down},
\]
with connecting edges. Let
\[
\Lambda^{j_0}_{j,k,n}
\]
denote the number of paths from the initial node $(0,j_0)$ to the
final node $(n,j)$ that has exactly $k$ upward jumps. The following
result may then also be of independent interest.

\begin{lemme}\label{lemme}
 Let $j_0\geq 0$ and $n\in\mathbb{N}$. Then
\[
\Lambda^{j_0}_{j,k,n} =
\begin{cases}
\displaystyle\binom{n}{k}, 
&\text{if } j=j_0+2k-n\\[-2mm]
&\text{and } n-\joinf\leq k\leq n,\\[2mm]
\displaystyle\binom{n}{k}-\binom{n}{k+\joinf+1},
&\text{if } j=j_0+2k-n\\[-2mm]
&\text{and } n-\lfloor\frac{n+j_0}{2}\rfloor\leq k\leq n-\joinf-1,\\[2mm]
\displaystyle\binom{n}{k-j}-\binom{n}{k-j-1},
&\text{if }0\leq j\leq n-\joinf-1\\[-2mm]
&\text{and } j\leq k\leq \lfloor\frac{n-\joinf-1+j}{2}\rfloor.
\end{cases}
\]
All the other values of $\Lambda^{j_0}_{j,k,n}$ are zero.
\end{lemme}

\begin{proof} (I) We will first consider the situation where $j_0$ is
  not an integer, see Figure~\ref{fig-cheuk}. In that case there are
  three disjoint classes of terminal nodes. The proof will be done in
  three steps devoted to the nodes in each of these classes.

  (1) The first class is given by the end-points of simple solid
  lines. These nodes are exactly those for which the set of
  predecessors (going back to $m=0$) is the same as in a
  \textit{traditional binomial tree}. It is then easy to count the
  number of paths. To arrive at such a node there are at least
  $k=n-\joinf$ and at most $k=n$ upward jumps, and there are
  $\binom{n}{k}$ paths with exactly $k$~upward jumps. The
  corresponding terminal levels are $j=j_0+k-(n-k)=j_0+2k-n$. Thus,
\[
\Lambda^{j_0}_{j,k,n}= \binom{n}{k}
\]
for $j=j_0+2k-n$ and $n-\joinf\leq k\leq n$. We note for later use
that the levels $j$ in this class range from $j_0+n-2\joinf$ to
$j_0+n$ in steps of 2, and that the number $k$ of upward jumps is
\begin{equation}\label{upwjumps}
k=\frac{n+j-j_0}{2}.
\end{equation}

(2) The second class is constituted of all the remaining terminal
nodes that belong to the traditional binomial tree starting at
$(0,j_0)$; in Figure~\ref{fig-cheuk} they are indicated by dashed
lines. We have here a \textit{partial binomial tree}: some of the
paths of the traditional binomial tree leading to such nodes have been
lost. Indeed, some paths would have had at least one intermediate node
at level zero. The possible terminal levels range from $\{j_0\}$ (if
$\joinf+n$ is even) or from $\{j_0\}+1$ (if $\joinf+n$ is odd) to
$j_0+n-2\joinf-2$ in steps of two; moreover there are such nodes only
when $n\geq \joinf+2$. We can rewrite the lower bound in both cases as
$j_0+n-2\lfloor \frac{n+j_0}{2}\rfloor$. By \eqref{upwjumps}, each
level corresponds to a unique number $k$ of upward jumps, and the
possible values of $k$ for nodes of the second class range from
$n-\lfloor \frac{n+j_0}{2}\rfloor$ to $n-\joinf-1$.

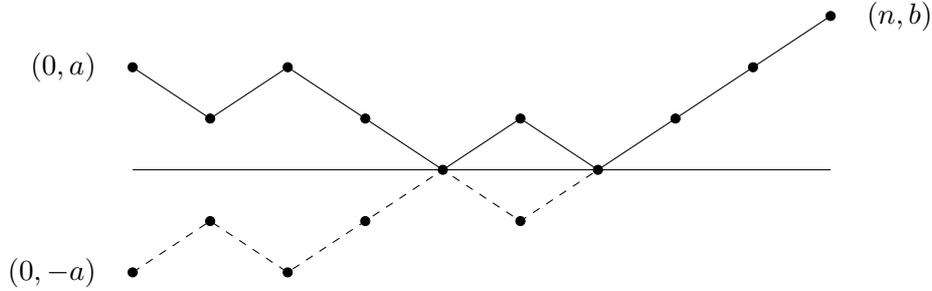
\begin{figure}[ht!]
\begin{center}
\begin{tikzpicture}[scale=0.68]
\draw (0,2) -- (1.5,1) -- (3,2) -- (4.5,1) -- (6,0) -- (7.5,1)-- (9,0)-- (10.5,1)-- (12,2)-- (13.5,3);
\draw (0,0) -- (13.5,0);
\draw[dashed] (0,-2) -- (1.5,-1) -- (3,-2) -- (4.5,-1) -- (6,0) -- (7.5,-1)-- (9,0);
\filldraw (0,2) circle(2.5pt);
\filldraw  (1.5,1) circle(2.5pt);
\filldraw (3,2)circle(2.5pt);
\filldraw  (4.5,1)circle(2.5pt);
\filldraw  (6,0)circle(2.5pt);
\filldraw  (7.5,1)circle(2.5pt);
\filldraw (9,0)circle(2.5pt);
\filldraw (10.5,1)circle(2.5pt);
\filldraw (12,2)circle(2.5pt);
\filldraw (13.5,3)circle(2.5pt);
\filldraw (0,-2)circle(2.5pt);
\filldraw (1.5,-1)circle(2.5pt);
\filldraw (3,-2)circle(2.5pt);
\filldraw (4.5,-1) circle(2.5pt);
\filldraw  (7.5,-1)circle(2.5pt);
\node[left] at (-0.5,2){$(0,a)$};
\node[left] at (-0.5,-2){$(0,-a)$};
\node[right] at (14,3){$(n,b)$};
\end{tikzpicture}
\caption{The reflection principle}
\label{fig-refl}
\end{center}
\end{figure}

Now, the number of paths leading to a node of the second class at
level~$j$ is again~$\binom{n}{k}$, but we have to withdraw all the
\textit{lost paths}. The number of lost paths is obtained by using the
reflection principle, see Figure \ref{fig-refl}. It states that there
is a one-to-one correspondence between the paths in a binomial tree
connecting the nodes~$(0,a)$ and~$(n,b)$ (with $a$,
$b\in\mathbb{N}_0$) in touching or crossing the $x$-axis and the paths
connecting the nodes~$(0,-a)$ and~$(n,b)$. Our case is equivalent to
the situation where $a=\joinf+1$ and $b=j-\{j_0\}+1=j-j_0+\joinf+1$.
It remains to count the number of paths from~$(0,-a)$ to~$(n,b)$. The
sum of the number of upward jumps~($U$) with the number of downward
jumps~($D$) is the number of periods~$n$, while the difference between
$U$ and $D$ is the overall increase, that is, $a+b$. We thus have that
$U=\frac{n+a+b}{2}$. Therefore, the number of lost paths is
\begin{equation*}
\binom{n}{\frac{n+a+b}{2}}
=\binom{n}{\frac{n+2\joinf+j-j_0+2}{2}}
=\binom{n}{k+\joinf+1},
\end{equation*}
where we have used  \eqref{upwjumps}. Thus,
\[
\Lambda^{j_0}_{j,k,n}= \binom{n}{k}\!-\!\binom{n}{k+\joinf +1}
\]
for $j=j_0+2k-n$ and $n-\lfloor \frac{n+j_0}{2}\rfloor\leq k \leq n-\joinf-1$.

(3) The third class of nodes is constituted by the remaining nodes;
they are indicated in Figure~\ref{fig-cheuk} by double lines (either
solid or dashdotted ones). They are the terminal nodes of a smaller
Cheuk-Vorst tree with the initial node at~$(\joinf+1,0)$. Moreover,
there are such nodes only when $n\geq \joinf+1$. The possible
terminal levels range from $j=0$ to $j=n-\joinf-1$.

Unlike in the two previous cases, the paths joining the initial node
$(0,j_0)$ with a terminal level~$j$ will not have a fixed number of
upward jumps, just like in the traditional Cheuk-Vorst tree. We will
prove that the number of paths with exactly $k$~upward jumps arriving
at a level~$j$ in the third class is given by
\begin{equation*}
\Lambda^{j_0}_{j,k,n}=\binom{n}{k-j}-\binom{n}{k-j-1}
\end{equation*}
for $0\leq j\leq n-\joinf-1$ and $j\leq k\leq
\lfloor\frac{n-\joinf-1+j}{2}\rfloor$, and that no other values of $k$
are possible.

We prove the claim by induction on~$n$, $n\geq \joinf+1$. The result
is trivial for $n=\joinf+1$. For $n=\joinf+2$, we can arrive at
level~$j=0$ (with $k=0$) or at $j=1$ (with $k=1$). In both cases,
there is exactly one possible path.

Now let $n\geq \joinf+3$. Again, for $j=n-\joinf-1$ (with
$k=n-\joinf-1$) and $j=n-\joinf-2$ (with $k=n-\joinf-2$) the result is
trivial. In both cases there is only one possible path, which are
respectively $\joinf+1$~downs followed by $n-\joinf-1$~ups and
$\joinf+2$~downs followed by $n-\joinf-2$~ups.

It remains to discuss the case when $0\leq j\leq n-\joinf-3$. If
$j=0$, there are two downward paths leading to it from period $n-1$
if~$n+\joinf$ is even, and three downward paths otherwise.
\begin{enumerate}
\item[(i)] If $n+\joinf$ is even, they come from the nodes located at
  levels~0 or~1 at period $n-1$.\\
  If $k=0$, the only possible path comes from $j=0$ and so
\[
\Lambda^{j_0}_{0,0,n}=\Lambda^{j_0}_{0,0,n-1}=1.
\]
If $1\leq k\leq \frac{n-\joinf-2}{2}$, then we have
\[
\begin{split}
  \Lambda^{j_0}_{0,k,n}&=\Lambda^{j_0}_{0,k,n-1}+\Lambda^{j_0}_{1,k,n-1}\\
  &=\binom{n-1}{k}-\binom{n-1}{k-1}+\binom{n-1}{k-1}-\binom{n-1}{k-2}\\
  &=\binom{n}{k}-\binom{n}{k-1}.
\end{split}
\]
There are no paths for~$k>\frac{n-\joinf-2}{2}$.
\item[(ii)] If $n+\joinf$ is odd, there are additional downward paths
  coming from the partial binomial tree. By \eqref{upwjumps}, these
  paths did exactly $\frac{n-\joinf-1}{2}$ ups before entering the
  Cheuk-Vorst tree. They make therefore no contribution
  if $0\leq k\leq \frac{n-\joinf-3}{2}$, and so we argue as in (i).\\
  If $k=\frac{n-\joinf-1}{2}$, it is impossible to come from
  level~$j=0$ at period $n-1$ and so we have with (2) that
\[
\begin{split}
  \Lambda^{j_0}_{0,k,n}&=\Lambda^{j_0}_{1,k,n-1}+\Lambda^{j_0}_{\{j_0\},k,n-1}\\
  &=\binom{n-1}{k-1}-\binom{n-1}{k-2}+\binom{n-1}{k}-\binom{n-1}{k+\joinf+1}\\
  &=\binom{n-1}{k-1}-\binom{n-1}{k-2}+\binom{n-1}{k}-\binom{n-1}{k-1}\\
  &=\binom{n}{k}-\binom{n}{k-1}.
\end{split}
\]
There are no paths for~$k>\frac{n-\joinf-1}{2}$.
\end{enumerate}

If $1\leq j\leq n-\joinf-3$, there are two paths leading to it from
period $n-1$, one upward path coming from the node located at
level~$j-1$ and one downward path coming from the node located at
level~$j+1$.  If $0\leq k<j$ there are no paths with $k$~ups. If
$k=j$, the only possible path comes from $j-1$ and so
\[
\Lambda^{j_0}_{j,j,n}=\Lambda^{j_0}_{j-1,j-1,n-1}=1.
\]
If $j+1\leq k \leq \lfloor \frac{n-\joinf-1+j}{2}\rfloor$, then we have
\[
\begin{split}
  \Lambda^{j_0}_{j,k,n}&=\Lambda^{j_0}_{j-1,k-1,n-1}+\Lambda^{j_0}_{j+1,k,n-1}\\
  &=\binom{n-1}{k-j}-\binom{n-1}{k-j-1}+\binom{n-1}{k-j-1}-\binom{n-1}{k-j-2}\\
  &=\binom{n}{k-j}-\binom{n}{k-j-1}.
\end{split}
\]
There are no paths for~ $k>\lfloor \frac{n-\joinf-1+j}{2}\rfloor$.

This proves the lemma in the case when $j_0$ is not an integer. 

\begin{figure}[ht!]
\begin{center}
\begin{tikzpicture}[scale=0.75]
\draw (0,1) -- (2,2) -- (4,3) -- (6,4) -- (8,5) -- (10,6);
\draw (0,1) -- (2,0) -- (4,1) -- (6,2) -- (8,3) -- (10,4);
\draw (2,2) -- (4,1);
\draw (4,3) -- (6,2);
\draw (6,4) -- (8,3);
\draw (8,5) -- (10,4);
\draw[dashed] (4,1) -- (6,0);
\draw[dashed] (6,2) -- (8,1);
\draw[dashed] (8,3) -- (10,2);
\draw[ultra thick] (6,0) -- (8,1) -- (10,2);
\draw[ultra thick] (8,1) -- (10,0);
\draw[double,dashdotted] (2,0) -- (4,0);
\draw[double] (4,0) -- (6,1) -- (8,2) -- (10,3);
\draw[double] (4,0) -- (6,0) -- (8,0) -- (10,0);
\draw[double] (6,1) -- (8,0);
\draw[double] (8,2) -- (10,1);
\draw[double] (8,0) -- (10,1);
\filldraw (0,1) circle(2.5pt);
\filldraw (2,2) circle(2.5pt);
\filldraw (2,0) circle(2.5pt);
\filldraw (4,3) circle(2.5pt);
\filldraw (4,1) circle(2.5pt);
\filldraw (6,4) circle(2.5pt);
\filldraw (6,2) circle(2.5pt);
\filldraw (6,0) circle(2.5pt);
\filldraw (8,5) circle(2.5pt);
\filldraw (8,3) circle(2.5pt);
\filldraw (8,1) circle(2.5pt);
\filldraw (10,6) circle(2.5pt);
\filldraw (10,4) circle(2.5pt);
\filldraw (10,2) circle(2.5pt);
\filldraw (10,0) circle(2.5pt);
\filldraw (4,0) circle(2.5pt);
\filldraw (6,1) circle(2.5pt);
\filldraw (8,0) circle(2.5pt);
\filldraw (8,2) circle(2.5pt);
\filldraw (10,0) circle(2.5pt);
\filldraw (10,1) circle(2.5pt);
\filldraw (10,2) circle(2.5pt);
\filldraw (10,3) circle(2.5pt);
\node[left] at (-0.5,1){$j_0=1$};
\node[right] at (10.5,6){$j=6$};
\node[right] at (10.5,4){$j=4$};
\node[right] at (10.5,3){$j=3$};
\node[right] at (10.5,2){$j=2$};
\node[right] at (10.5,2){$j=2$};
\node[right] at (10.5,1){$j=1$};
\node[right] at (10.5,0){$j=0$};
\node[below] at (0,-0.25){$m=0$};
\node[below] at (2,-0.25){$m=1$};
\node[below] at (4,-0.25){$m=2$};
\node[below] at (6,-0.25){$m=3$};
\node[below] at (8,-0.25){$m=4$};
\node[below] at (10,-0.25){$m=5$};
\end{tikzpicture}
\caption{CV tree for the call with $n=5$ and $j_0=1$}
\label{fig-cheuk-int}
\end{center}
\end{figure}
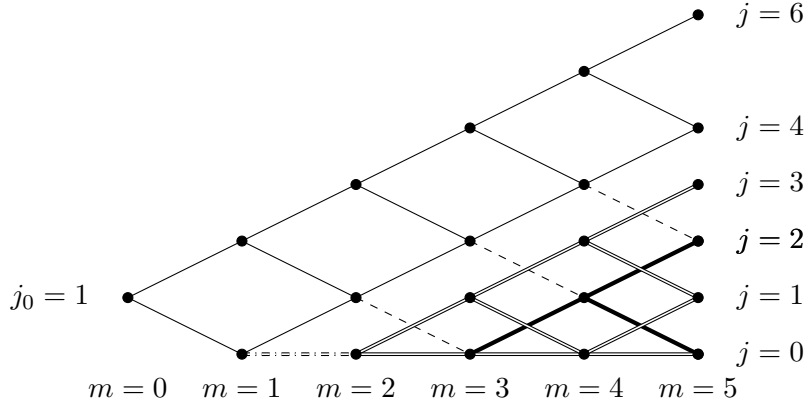

(II) The situation changes slightly when $j_0$ is an integer, see
Figure \ref{fig-cheuk-int}.  The three classes of nodes are defined as
before: Those for which the set of predecessors is the same as in the
traditional binomial tree starting from $(0,j_0)$ (indicated by solid
lines); the remaining terminal nodes of the traditional binomial tree
(indicated by dashed lines or by thick solid lines); and the terminal
nodes of the smaller Cheuk-Vorst tree starting at $(j_0+1,0)$
(indicated by double lines, thick solid lines, or the dash-dotted
line). Note, however, that some terminal nodes belong to both the
second and the third class.

(1) For the first class, we apply the same argument as before.

(2) For the nodes in the second class we only count the paths that are
inside the binomial tree. As before, there are such nodes only when
$n\geq j_0+2$, the possible number $k$ of upward paths ranges from
$n-\lfloor \frac{n+j_0}{2}\rfloor$ to $n-j_0-1$, and the terminal
level is given by $j=j_0+2k-n$. The number of corresponding paths
(inside the binomial tree) is $\binom{n}{k}$ minus the number of lost
paths. But now the lost paths are those paths in the binomial tree
that hit level $-1$ at some point. This is equivalent to counting the
paths from $(0,j_0+1)$ to $(n,j+1)$ that hit level 0, and the number
of such paths is, as before, $\binom{n}{k+j_0+1}$. This confirms the
lemma for the stated values of $j$ and $k$.

(3) Finally, for the nodes of the third class we count those paths
that at least once follow a path outside the binomial tree (in this
way we do not count paths twice when the node also belongs to the
second class). The argument as before confirms the third alternative
in the lemma.

We note that if the terminal level $j$ belongs to both the second and
the third class then the paths inside the binomial tree have exactly
$k=\frac{n+j-j_0}{2}$ upward jumps (whence $n+j-j_0$ is even), while
the paths outside this tree have at most
\[
\Big\lfloor\frac{n-j_0-1+j}{2}\Big\rfloor < \frac{n+j-j_0}{2}
\]
upward jumps. This shows that for such levels $j$ there is no conflict
in the statement of the lemma.
\end{proof}

Combining the formula \eqref{price} for the price of a lookback call
with Lemma \ref{lemme} we obtain the following.

\begin{theorem}\label{t-call}
  Let $0\leq t< T$ and $n\in\mathbb{N}$. The price of a European
  lookback call option with floating strike at time $t$ is given by
\begin{equation}\label{somme-call}
C^{fl}_n(t) = S_{t} (V_1-V_2+V_3),
\end{equation}
where
\begin{align*}
V_1&=\sum_{k=k_{\min}}^{n}(1-u^{n-j_0-2k}) \,\binom{n}{k}q^k\,(1-q)^{n-k},\\
V_2&=\sum_{k=k_{\min}}^{n-\joinf-1}\!\!\!(1-u^{n-j_0-2k})\,\binom{n}{k+\joinf+1}
q^k(1-q)^{n-k},\\
V_3&=\sum_{j=0}^{n-\joinf-1}\!\!\!(1-u^{-j}) \sum_{k=j}^{k_{\max}}
\Big[\binom{n}{k-j}-\binom{n}{k-j-1}\Big] \,q^k(1-q)^{n-k}
\end{align*}
with $j_0$ given by \eqref{j0}, $q$ given by \eqref{qu},
$k_{\min}=n-\lfloor \frac{n+j_0}{2}\rfloor$ and
$k_{\max}=\lfloor\frac{n-\joinf-1+j}{2}\rfloor$.
\end{theorem}

\begin{figure}[ht!]
\begin{center}
\begin{tikzpicture}[scale=1.5]
\draw (0,0) -- (1,0) -- (2,0) -- (3,0);
\draw (0,0) -- (1,-0.5) -- (2,-1) -- (3,-1.5);
\draw (1,0) -- (2,-0.5) -- (3,-1);
\draw (1,-0.5) -- (2,0) -- (3,-0.5);
\draw (2,-0.5) -- (3,0);
\draw (2,-1) -- (3,-0.5);
\filldraw (0,0) circle(1.25pt);
\filldraw (1,0) circle(1.25pt);
\filldraw (1,-0.5) circle(1.25pt);
\filldraw (2,0) circle(1.25pt);
\filldraw (2,-0.5) circle(1.25pt);
\filldraw (2,-1) circle(1.25pt);
\filldraw (3,0) circle(1.25pt);
\filldraw (3,-0.5) circle(1.25pt);
\filldraw (3,-1) circle(1.25pt);
\filldraw (3,-1.5) circle(1.25pt);
\end{tikzpicture}
\caption{CV tree for the put with $n=3$}
\label{fig-cheuk-put}
\end{center}
\end{figure}
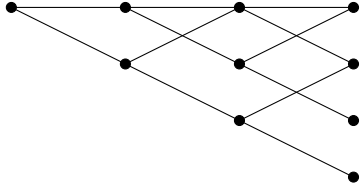

In the same way one can determine the price of a lookback put option
with floating strike. In that case we would have negative values for
levels because the underlying price is always lower than the maximal
price, see Figure~\ref{fig-cheuk-put}. For convenience we will write
these levels as $-j$ with $j\geq 0$. In terms of the levels~$j$ we are
then led to the same tree as in the case of the call.

With this adaptation, the level for the put at time $t_m$ is given by
a non-negative number~$j$ such that
\[
S_{t_m}=\big(\max_{t^*\leq t_m}S_{t^*}\big)u^{-j}.
\]
The initial level~$j_0$ (at initial time~$t$) of the tree satisfies
\begin{equation*}
S_{t}=\big(\max_{t^*\leq t}S_{t^*}\big)u^{-j_0},
\end{equation*}
so that
\begin{equation}\label{j0_put}
j_0=\frac{\log (\max_{t^*\leq t}S_{t^*}/S_{t})}{\sigma\sqrt{\tau/n}}.
\end{equation}

It can be shown as in Cheuk and Vorst \cite{cheuk_vorst} that the
price at time~$t$ of the European lookback put option with floating
strike is given by
\begin{equation*}
P^{fl}_n(t) = S_{t} \sum_{j\in J} (u^{j}-1)\sum_{k=0}^n
\Lambda^{j_0}_{j,k,n} (1-q)^{k}\,q^{n-k},
\end{equation*}
where $J$ is the set of possible levels (taken positively) at
maturity, $q$ is given by \eqref{qu} and $\Lambda^{j_0}_{j,k,n}$ is
the number of paths from the initial node $(0,j_0)$ to level~$j$ at
maturity that have exactly $k$ upward jumps.  Since we have evaluated
these numbers of paths in Lemma \ref{lemme} we obtain the following.

\begin{theorem}\label{t-put}
  Let $0\leq t< T$ and $n\in\mathbb{N}$. The price of a European
  lookback put option with floating strike at time~$t$ is given by
\[
P^{fl}_n(t) = S_{t} (V_1-V_2+V_3),
\]
where
\begin{align*}
V_1&=\sum_{k=k_{\min}}^{n}(u^{j_0+2k-n}-1) \,\binom{n}{k}(1-q)^{k}\,q^{n-k},\\
V_2&=\sum_{k=k_{\min}}^{n-\joinf-1} (u^{j_0+2k-n}-1)\,
	\binom{n}{k+\joinf+1} 	(1-q)^{k}\,q^{n-k},\\
V_3&=\sum_{j=0}^{n-\joinf-1}\!\!\!(u^{j}-1) \sum_{k=j}^{k_{\max}}
\Big[\binom{n}{k-j}-\binom{n}{k-j-1}\Big] \,(1-q)^{k}\,q^{n-k}
\end{align*}
with $j_0$ given by \eqref{j0_put}, $q$ given by \eqref{qu},
$k_{\min}=n-\lfloor \frac{n+j_0}{2}\rfloor$ and
$k_{\max}=\lfloor\frac{n-\joinf-1+j}{2}\rfloor$.
\end{theorem}

%%%%%%%%%%%%%%%%%%%%%%%%%%%%%%%%%%%%%%%%%%%%%%%%%%%%%%%%%%%%%%%%%%%%%%%%%%%
\section{Asymptotic expansions of the binomial cumulative distribution
  function}
%%%%%%%%%%%%%%%%%%%%%%%%%%%%%%%%%%%%%%%%%%%%%%%%%%%%%%%%%%%%%%%%%%%%%%%%%%%

In order to obtain asymptotic expansions for the price of lookback
options as the number of steps $n$ tends to infinity we first need to
derive an asymptotic expansion of the binomial cumulative distribution
function
\[
\bin{n,p_n}(j_n) = \sum_{k=0}^{j_n} \binom{n}{k}p_n^k (1-p_n)^{n-k}
\]
or, equivalently, for the complementary binomial cumulative
distribution function
\[
\binc{n,p_n}(j_n) = \sum_{k=j_n+1}^{n} \binom{n}{k}p_n^k (1-p_n)^{n-k},
\]
and that with a lower error term than in the expansions found in the
literature. Moreover, our method allows us to treat any sequence
$(p_n)$ satisfying the rather weak assumption that
\[
0<\liminf_{n\to\infty} p_n\leq \limsup_{n\to\infty} p_n<1.
\]

Our approach relies on the work of Uspensky~\cite{uspensky} who
represented binomial probabilities in a convenient analytical
form. Chang and Palmer~\cite[Lemma 1]{chang_palmer} refined this
result in order to obtain an approximation of order $o(n^{-1})$ in the
case when $p_n\to 1/2$. Recently, Lin and Palmer~\cite[Lemma
C.1]{lin_palmer} provided an estimate of order $O(n^{-2})$ if $p_n =
1/2+O(n^{-1/2})$.
 
Chang and Palmer~\cite{chang_palmer} have used their approximation in
order to obtain an asymptotic expansion for digital and European
options with a remainder term $O(n^{-3/2})$ (see \cite[Theorem
1.1]{lin_palmer}), while Lin and Palmer~\cite{lin_palmer} treated
barrier options with the same precision. The second
author~\cite{heuw1} recently evaluated lookback options with floating
strike, limiting himself to the price at emission (that is, at
$t=0$). Moreover, when the spot rate $r$ equals 0 then he only
obtained a remainder term $O(n^{-1})$ (see \cite[Remark 4.1]{heuw1}).

In order to obtain an asymptotic expansion for lookback options with
floating strike with a remainder term $O(n^{-3/2})$, and that for any
spot rate $r\geq 0$ and for the price at any time $t\geq 0$, we will
need an approximation of the binomial cumulative distribution with
remainder term $O(n^{-5/2})$, as provided by the following theorem.

In the sequel, $C>0$ denotes a generic constant, which may have a
different value at each occurrence. Note also that, for the sake of
readability, we will often drop the index $n$.

Let $\Phi$ denote the standard normal cumulative distribution
function.

\begin{theorem}\label{thm-general}
Suppose that $p=p_n$ satisfies
\[
0<\liminf_{n\to\infty} p_n\leq \limsup_{n\to\infty} p_n<1.
\]
If $q=q_n=1-p_n$ and $0\leq j=j_n\leq n$, then
\[
\begin{split}
 &\sum_{k=0}^j \binom{n}{k}p^k q^{n-k}\\
  &\quad=\Phi(y)+ \frac{e^{-\frac{1}{2}y^2}}{\sqrt{2\pi}}
  \Big(\frac{P_1}{\sqrt{V}}+\frac{P_2}{V}
  +\frac{P_3}{V^{3/2}}+\frac{P_4}{V^2}\Big)
  +O\Big(\frac{1}{n^{5/2}}\Big)
\end{split}
\]
as $n\to\infty$, where $V=npq$, $y=\frac{j-np+1/2}{\sqrt{V}}$ and
\begin{itemize}
\setlength{\itemsep}{1mm}
\item[$P_1=$\!] $\frac{1}{6}(q-p)(1-y^2)$,
\item[$P_2=$\!] $y[\frac{1}{72}(-3+7y^2-y^4)
	-\frac{pq}{36}(-3+11y^2-2y^4)]$,
\item[$P_3=$\!] $(q-p)[
	\frac{1}{6480}(123+129y^2-384y^4+95y^6-5y^8)\\ 
	-\frac{pq}{3240}(3+69y^2-399y^4+145y^6-10y^8)]$,
\item[$P_4=$\!] $y[
	\frac{1}{155520}(-4293-1359y^2+6165y^4-1971y^6+185y^8
		-5y^{10})\\
	+\frac{pq}{38880}(3105+1395y^2-7794y^4+2979y^6
		-325y^8+10y^{10})\\
	+\frac{p^2 q^2}{38880}(135-1035y^2+7947y^4-4167y^6
		+560y^8-20y^{10})]$.
\end{itemize}
\end{theorem}

As mentioned above, the proof of this result will be based on an
analytical representation of binomial probabilities due to
Uspensky~\cite[p. 121]{uspensky}.

\begin{theorem*}[Uspensky] 
Let $0<p<1$, $q=1-p$ and $0\leq j\leq n$ be fixed numbers. 
Then 
\[
\sum_{k=0}^j \binom{n}{k}p^k q^{n-k} = J(y)-J(y'),
\]
where 
\[
y=\frac{j-np+\frac{1}{2}}{\sqrt{V}}\quad\text{and}\quad
y'=-\frac{np+\frac{1}{2}}{\sqrt{V}}
\]
with $V=npq$; here the function $J$ is defined by 
\begin{equation}\label{fct-J}
  J(y)=\frac{1}{2\pi}\int_0^{\pi}\rho^n\,
  \frac{\sin(y\sqrt{V}\varphi-\chi)}
  {\sin\frac{\varphi}{2}}\,\mathrm{d}\varphi,\quad y\in\mathbb{R},
\end{equation}
where 
\[
\rho= |pe^{i\varphi}+q|,\quad\omega = \arg (pe^{i\varphi}+q)\quad\text{and}\quad \chi=n\omega-np\varphi.
\]
\end{theorem*}

We will derive Theorem~\ref{thm-general} from Uspensky's
representation by employing and refining his ideas
(see~\cite[pp. 121--129]{uspensky}). For this we will need two
preliminary lemmas.

\begin{lemme}\label{lemme-da-R}
Suppose that $p=p_n$ satisfies
\[
0<\liminf_{n\to\infty} p_n\leq \limsup_{n\to\infty} p_n<1.
\]
Let $q=q_n=1-p_n$ and $V=npq$. For a fixed constant $M>0$, let
$\varphi$ be a positive number such that $\varphi\leq M/V^{1/4}$. Then
we have for $R= |pe^{i\varphi}+q|^n$ that
\begin{equation*}
R= e^{R_2\varphi^2}\Big(1+R_4\varphi^4+R_6\varphi^6
+\frac{1}{2}R_4^2\varphi^8
+\sum_{k=1}^3 O(n^k\varphi^{2k+6})\Big)
\end{equation*}
as $n\to\infty$, where
\begin{itemize}
\setlength{\itemsep}{1mm}
\item[$R_2=$\!] $-\frac{1}{2}V$,
\item[$R_4=$\!] $\frac{1}{4}V\big(\frac{1}{6}-pq\big)$,
\item[$R_6=$\!] $-\frac{1}{6}V\big(\frac{1}{120}-\frac{1}{4}pq+p^2q^2\big)$.
\end{itemize}
\end{lemme}

\begin{proof}
Let
\[
\rho = |pe^{i\varphi}+q|,
\]
so that $R=\rho^n$. We then have that
\[
\rho=(p^2+2pq\cos\varphi+q^2)^{1/2}
=\Big(1-4pq\sin^2\frac{\varphi}{2}\Big)^{1/2}.
\]
This gives us that
\begin{equation}\label{lemme-lnrho}
\begin{split}
\log\rho
&=\frac{1}{2}\log\Big(1-4pq\sin^2\frac{\varphi}{2}\Big)\\
&=-2pq\sin^2\frac{\varphi}{2}
	-4(pq)^2\sin^4\frac{\varphi}{2}
	-\frac{32}{3}(pq)^3\sin^6\frac{\varphi}{2}
	-\delta,
\end{split}
\end{equation}
where $\delta=32(1-\eta)^{-4}(pq)^4\sin^8\frac{\varphi}{2}$ for some
real number~$\eta$ between 0 and $4pq\sin^2\frac{\varphi}{2}$. Since
$\liminf_{n\to\infty} p_n, \liminf_{n\to\infty} q_n >0$ we have that
$V=npq\to\infty$ and therefore $0\leq \varphi\leq M/V^{1/4}\to 0$ as
$n\to\infty$. Thus
\begin{equation}\label{inegalite-delta}
0\leq\delta\leq C p^4q^4\varphi^8
\end{equation}
for $n$~sufficiently large. 

Now, in order to obtain bounds for $\log\rho$, we use suitable Taylor
expansions about~0. First, we have
\begin{equation*}
\sin^2\frac{\varphi}{2}=\frac{1}{4}\varphi^2
	-\frac{1}{48}\varphi^4
	+\frac{1}{1440}\varphi^6
	-\frac{1}{80640}\varphi^8
	+\frac{1}{7257600}\varphi^{10}
	+o(\varphi^{10}).
\end{equation*} 
Thus we have
\begin{equation}\label{inegalite-sin2}
\frac{1}{4}\varphi^2
	-\frac{1}{48}\varphi^4
	+\frac{1}{1440}\varphi^6
	-\frac{1}{80640}\varphi^8
\leq\sin^2\frac{\varphi}{2}
\leq\frac{1}{4}\varphi^2
	-\frac{1}{48}\varphi^4
	+\frac{1}{1440}\varphi^6
\end{equation}
for $n$~sufficiently large. Next we deduce from
\begin{equation*}
\sin^4\frac{\varphi}{2}=\frac{1}{16}\varphi^4
	-\frac{1}{96}\varphi^6
	+\frac{1}{1280}\varphi^8
	-\frac{17}{483840}\varphi^{10}
	+o(\varphi^{10})
\end{equation*}
that
\begin{equation}\label{inegalite-sin4}
\frac{1}{16}\varphi^4
	-\frac{1}{96}\varphi^6
\leq\sin^4\frac{\varphi}{2}
\leq\frac{1}{16}\varphi^4
	-\frac{1}{96}\varphi^6
	+\frac{1}{1280}\varphi^8
\end{equation}
for $n$~sufficiently large. Finally, we obtain from
\begin{equation*}
\sin^6\frac{\varphi}{2}=\frac{1}{64}\varphi^6
	-\frac{1}{256}\varphi^8
	+\frac{7}{15360}\varphi^{10}
	+o(\varphi^{10})
\end{equation*}
that
\begin{equation}\label{inegalite-sin6}
\frac{1}{64}\varphi^6
	-\frac{1}{256}\varphi^8
\leq \sin^6\frac{\varphi}{2}
\leq\frac{1}{64}\varphi^6
\end{equation}
for $n$~sufficiently large.  From (\ref{inegalite-delta}),
(\ref{inegalite-sin2}), (\ref{inegalite-sin4}),
(\ref{inegalite-sin6}), we find upper and lower bounds for the
expression~(\ref{lemme-lnrho}):
\begin{equation}\label{lemme-upperbond}
\begin{split}
\log\rho 
&\leq -2pq\Big(\frac{1}{4}\varphi^2
		-\frac{1}{48}\varphi^4
		+\frac{1}{1440}\varphi^6
		-\frac{1}{80640}\varphi^8\Big)\\
&\qquad -4p^2 q^2\Big(\frac{1}{16}\varphi^4
		-\frac{1}{96}\varphi^6\Big)
	-\frac{32}{3}p^3 q^3\Big(\frac{1}{64}\varphi^6
		-\frac{1}{256}\varphi^8\Big)\\
&=\frac{R_2}{n}\varphi^2+\frac{R_4}{n}\varphi^4
+\frac{R_6}{n}\varphi^6+\frac{1}{24}pq\Big(\frac{1}{1680}+p^2q^2\Big)\varphi^8
\end{split}
\end{equation}
and
\begin{equation}\label{lemme-lowerbond}
\begin{split}
\log\rho 
&\geq -2pq\Big(\frac{1}{4}\varphi^2
		-\frac{1}{48}\varphi^4+\frac{1}{1440}\varphi^6\Big)
-4p^2 q^2\Big(\frac{1}{16}\varphi^4
		-\frac{1}{96}\varphi^6\\
&\qquad+\frac{1}{1280}\varphi^8\Big)-\frac{32}{3}p^3 q^3
\Big(\frac{1}{64}\varphi^6\Big)
	-C p^4 q^4\varphi^8\\
&= \frac{R_2}{n}\varphi^2+\frac{R_4}{n}\varphi^4
+\frac{R_6}{n}\varphi^6
-p^2q^2\Big(\frac{1}{320}+Cp^2 q^2\Big)\varphi^8.
\end{split}
\end{equation}
By combining (\ref{lemme-upperbond}) and (\ref{lemme-lowerbond}), we
establish that
\begin{equation*}
e^{R_2\varphi^2+R_4\varphi^4+R_6\varphi^6}e^{\Delta_1}\leq R=\rho^n\leq
e^{R_2\varphi^2+R_4\varphi^4+R_6\varphi^6}e^{\Delta_2},
\end{equation*}
where
\[
\Delta_1=-V pq\Big(\frac{1}{320}+C p^2q^2\Big)\varphi^8\leq 0,\quad
\Delta_2=\frac{1}{24}V\Big(\frac{1}{1680}+p^2q^2\Big)\varphi^8\geq 0.
\]
This clearly implies that
\begin{equation*}
\big|R-e^{R_2\varphi^2+R_4\varphi^4+R_6\varphi^6}\big|
\leq e^{R_2\varphi^2+R_4\varphi^4+R_6\varphi^6}
\big(e^{\Delta_2}-e^{\Delta_1}\big).
\end{equation*}
Using the facts that $e^x\leq 1+2x$ for $0\leq
x\leq 1$ and $e^{x}\geq 1+x$, we have for large $n$ that
\[
e^{\Delta_2}-e^{\Delta_1}\leq 2\Delta_2-\Delta_1
\leq C V \varphi^8;
\]
note that $\Delta_2\to 0$ as $n\to\infty$. It follows that
\begin{equation*}
\big|R-e^{R_2\varphi^2+R_4\varphi^4+R_6\varphi^6}\big|
\leq C e^{R_2\varphi^2+R_4\varphi^4+R_6\varphi^6} V \varphi^8,
\end{equation*}
in other words,
\begin{equation}\label{lemme-resultat-interm}
R=e^{R_2\varphi^2+R_4\varphi^4+R_6\varphi^6}+ O(e^{R_2\varphi^2+R_4\varphi^4+R_6\varphi^6} V \varphi^8).
\end{equation}
Now, the definitions of $R_4$ and $R_6$ and the fact that
$0\leq\varphi\leq M/V^{1/4}$ imply that
\begin{equation}\label{lemme-r4r6}
R_4\varphi^4+R_6\varphi^6\leq C.
\end{equation}
By combining (\ref{lemme-resultat-interm}) and (\ref{lemme-r4r6}), we
see that
\begin{equation}\label{lemme-resultat-interm2}
R=e^{R_2\varphi^2}(e^{R_4\varphi^4+R_6\varphi^6}+O(n\varphi^8)).
\end{equation}

We rewrite the second exponential term with a Taylor expansion about~0
to obtain for some $\eta$ between 0 and $R_4\varphi^4+R_6\varphi^6$,
\begin{equation*}
e^{R_4\varphi^4+R_6\varphi^6}
= 1+(R_4\varphi^4+R_6\varphi^6)+\frac{1}{2}(R_4\varphi^4+R_6\varphi^6)^2
	+\frac{1}{6}e^{\eta}(R_4\varphi^4+R_6\varphi^6)^3,
\end{equation*}
hence with \eqref{lemme-r4r6}
\begin{equation}\label{lemme-R2}
e^{R_4\varphi^4+R_6\varphi^6}=1+R_4\varphi^4+R_6\varphi^6
	+\frac{1}{2}R_4^2\varphi^8+O(n^2\varphi^{10})
	+O(n^3\varphi^{12}).
\end{equation}
Combining~(\ref{lemme-resultat-interm2}) and~(\ref{lemme-R2}) we
obtain the result.
\end{proof}

\begin{lemme}\label{lemme-ebnphi}
Suppose that $p=p_n$ satisfies
\[
0<\liminf_{n\to\infty} p_n\leq \limsup_{n\to\infty} p_n<1.
\]
Let $q=q_n=1-p_n$ and $V=npq$. Then, for any non-negative integer~$m$,
there exists a constant $C>0$ such that
\[
I_m:=\int_{0}^{\infty}e^{-\frac{1}{2}V\varphi^2}\varphi^m\,\mathrm{d}\varphi
\leq C n^{-(m+1)/2}.
\]

Moreover, suppose that $\tau=\tau_n>0$ satisfies
$\tau^{-1}=O(n^{\alpha})$ with $\alpha<1/2$. Then for any
integers~$m$ and~$k$, there exists a constant $C>0$ such that
\[
I^*_m:=\int_{\tau}^{\infty}e^{-\frac{1}{2}V\varphi^2}\varphi^m\,\mathrm{d}\varphi
\leq C n^{-k}.
\]
\end{lemme}

\begin{proof}
  We first note that by our hypotheses there is some $\delta>0$ such
  that, for large $n$,
\begin{equation}\label{lowerest}
V\geq \delta n\quad\text{and}\quad \tau\sqrt{V}\geq \delta n^{1/2-\alpha}.
\end{equation}
We set $x=\sqrt{V}\varphi$. Then the integrals become
\begin{equation*}
I_m=V^{-(m+1)/2}\int_{0}^{\infty}e^{-\frac{1}{2}x^2}x^m\,\mathrm{d}x
\end{equation*}
and
\begin{equation*}
I_m^*=V^{-(m+1)/2}\int_{\tau\sqrt{V}}^{\infty}e^{-\frac{1}{2}x^2}x^m\,\mathrm{d}x.
\end{equation*}
The integrands can be bounded by a function~$C e^{-x}$ if $m\geq 0$,
or if $m<0$ and $x\geq 1$. Since, by \eqref{lowerest},
$\tau\sqrt{V}\to\infty$ as $n\to\infty$, we obtain that
\begin{equation*}
I_m\leq C V^{-(m+1)/2}
\end{equation*}
and
\begin{equation*}
I_m^*\leq C V^{-(m+1)/2}e^{-\tau\sqrt{V}}.
\end{equation*}
Thus we deduce with \eqref{lowerest} that
\begin{equation*}
I_m\leq C n^{-(m+1)/2}
\end{equation*}
and, for any~$k$,
\begin{equation*}
I_m^*\leq C n^{-k}.
\end{equation*}
\end{proof}

We can now prove our main result.

\begin{proof}[Proof of Theorem~\ref{thm-general}]
  The proof will be divided into several steps.

  (1) In view of Uspensky's representation stated above the value
  $J(y)$ plays a crucial role, see \eqref{fct-J}. We will split the
  integral into two parts. Let
\[
\tau=V^{-1/4},
\]
so that $\tau^{-1}=O(n^{1/4})$.  Throughout the proof we suppose
that~$n$ is large enough to have that $\tau\leq \frac{\pi}{2}$. Then
\begin{equation*}
J^{*}(y):=\Big|\frac{1}{2\pi}\int_{\tau}^{\pi} 
	\rho^n\,\frac{\sin(y\sqrt{V}\varphi-\chi)}
	{\sin\frac{\varphi}{2}}\,\mathrm{d}\varphi\Big|
\leq \frac{1}{2\pi}\int_{\tau}^{\pi} \frac{\rho^n}{\sin\frac{\varphi}{2}}
	\,\mathrm{d}\varphi.
\end{equation*}
With~(\ref{lemme-lnrho}) we have that
\begin{equation*}\label{ineg-nlogrho}
n\log\rho\leq -2V\sin^2\frac{\varphi}{2}.
\end{equation*}
Applying the fact that $\sin\frac{\varphi}{2}\geq\frac{\varphi}{\pi}$
for $0\leq\varphi\leq\pi$, we obtain that
\begin{equation*}
  J^{*}(y)\leq\frac{1}{2\pi}\int_{\tau}^{\pi}\frac{e^{-2V\sin^2\frac{\varphi}{2}}}{\sin\frac{\varphi}{2}}
\,\mathrm{d}\varphi
  \leq\frac{1}{2}\int_{\tau}^{\pi}e^{-\frac{2V}{\pi^2}\varphi^2}
  \frac{1}{\varphi}\,\mathrm{d}\varphi
  \leq\frac{1}{2}\int_{\tau}^{\infty}e^{-\frac{2V}{\pi^2}\varphi^2}
  \frac{1}{\varphi}\,\mathrm{d}\varphi.
\end{equation*}
Substituting $x=\frac{2}{\pi}\varphi$ and applying
Lemma~\ref{lemme-ebnphi}, we have that
\begin{equation}\label{approx-int-Jstar}
  J^{*}(y)\leq\frac{1}{2}\int_{\frac{2}{\pi}\tau}^{\infty}
  e^{-\frac{1}{2}Vx^2}\frac{1}{x}\,\mathrm{d}x
  =O\Big(\frac{1}{n^{5/2}}\Big)
\end{equation}
since $(\frac{2}{\pi}\tau)^{-1}=O(n^{1/4})$. From~(\ref{fct-J})
and~(\ref{approx-int-Jstar}) we then obtain that
\begin{equation}\label{fct-J-zero-tau}
  J(y)=\frac{1}{2\pi}\int_0^{\tau}\rho^n\,
  \frac{\sin(y\sqrt{V}\varphi-\chi)}
  {\sin\frac{\varphi}{2}}\,\mathrm{d}\varphi+O\Big(\frac{1}{n^{5/2}}\Big).
\end{equation}

(2) Looking at the integrand of $J(y)$ we now want to estimate
$\sin(a-\chi)$ in powers of~$\varphi$, where $a\in\mathbb{R}$; recall
that
\[
\omega = \arg (pe^{i\varphi}+q)\quad\text{and}\quad \chi=n\omega-np\varphi.
\]
By Taylor expansion we have that
\begin{equation}\label{lemme-sin1}
\begin{split}
\sin(a-\chi)&=
\sin(a)
-\cos(a)\chi
-\frac{1}{2}\sin(a)\chi^2
+\frac{1}{6}\cos(a)\chi^3\\
&\quad +\frac{1}{24}\sin(a)\chi^4
-\frac{1}{120}\cos(a-\eta)\chi^5
\end{split}
\end{equation}
for some $\eta$ between 0 and $\chi$. Since by~\eqref{fct-J-zero-tau}
it suffices to assume that $0\leq \varphi\leq \tau\leq\frac{\pi}{2}$
we see that
\[
\omega=\arctan\frac{p\sin\varphi}{p\cos\varphi+q},
\]
so that
\[
\chi= n\arctan\frac{p\sin\varphi}{p\cos\varphi+q} -np\varphi.
\]
Now,
\[
\frac{d}{d\varphi}\Big(\frac{\chi(\varphi)}{n}\Big) = \frac{1}{2}-p + \frac{p-q}{2}\frac{1}{1+2pq(\cos\varphi-1)},
\]
where we have used that $p+q=1$. Note that, as $x\to 0$,
\[
\frac{1}{1+2pqx} = 1-2pqx+(2pq)^2x^2 +O(x^3),
\]
where the constant in the big-O condition does not depend on $n$ since
$pq$ is bounded in $n$. Thus the Taylor expansion
\[
\cos\varphi - 1 = -\frac{\varphi^2}{2}+\frac{\varphi^4}{24} + O(\varphi^6)
\]
gives us that
\begin{align*}
\frac{d}{d\varphi}\Big(\frac{\chi(\varphi)}{n}\Big) &= \frac{1}{2}-p + \frac{p-q}{2}\Big(1-2pq\Big(-\frac{\varphi^2}{2}+\frac{\varphi^4}{24}\Big) + (2pq)^2\frac{\varphi^4}{4} + O(\varphi^6)\Big)\\
&= \frac{1}{2} pq (p-q) \varphi^2 - \frac{1}{24}pq(p-q)(1-12pq)\varphi^4+ O(\varphi^6)
\end{align*}
and hence
\begin{equation}\label{lemme-chi}
\chi =\chi_3\varphi^3+\chi_5\varphi^5+nO(\varphi^7)
\end{equation}
with
\begin{itemize}
\setlength{\itemsep}{1mm}
\item[$\chi_3=$\!] $\frac{1}{6}V(p-q)$,
\item[$\chi_5=$\!] $-\frac{1}{120}V(p-q)(1-12pq)$.
\end{itemize}
Applying~(\ref{lemme-sin1}) and~(\ref{lemme-chi}) we obtain
\begin{equation*}
\begin{split}
\sin(a\!-\!\chi)\!&=\!
	\sin(a)
	\!-\!\cos(a)
		(\chi_3\varphi^3\!+\!\chi_5\varphi^5\!+\!nO(\varphi^7))\\
&\quad -\!\frac{1}{2}\sin(a)(\chi_3\varphi^3\!+\!\chi_5\varphi^5\!+\!nO(\varphi^7))^2
	\!+\!\frac{1}{6}\cos(a)(\chi_3\varphi^3\!
	+\!nO(\varphi^5))^3\\
&\quad +\!\frac{1}{24}\sin(a)(\chi_3\varphi^3\!
	+\!nO(\varphi^5))^4
	\!-\!\frac{1}{120}\cos(a\!-\!\eta)(nO(\varphi^3))^5,
\end{split}
\end{equation*}
and hence
\begin{equation}\label{lemme-sin2}
\begin{split}
\sin(a\!-\!\chi)\!&=\!
\sin(a)\!-\!\cos(a)\chi_3\varphi^3
	\!-\!\cos(a)\chi_5\varphi^5
	\!-\!\frac{1}{2}\sin(a)\chi_3^2\varphi^6\\
&\quad \!-\!\sin(a)\chi_3 \chi_5\varphi^8
	\!+\!\frac{1}{6}\cos(a)\chi_3^3\varphi^9
	\!+\!\frac{1}{24}\sin(a)\chi_3^4\varphi^{12}\\
&\quad +\sum_{k=1}^5 O(n^k\varphi^{2k+5}).
\end{split}
\end{equation}

(3) Next, using Laurent expansion, we have that
\begin{equation*}
\frac{1}{\sin\frac{\varphi}{2}}=\frac{2}{\varphi}+\frac{1}{12}\varphi
	+\frac{7}{2880}\varphi^3+O(\varphi^5),
\end{equation*}
which together with Lemma \ref{lemme-da-R} and the fact that
$R=\rho^n$ gives that
\begin{equation}\label{lemme-rsin}
\begin{split}
\frac{\rho^n}{\sin\frac{\varphi}{2}}&=
e^{R_2\varphi^2}\Big(\frac{2}{\varphi}
	+\frac{1}{12}\varphi
	+\Big[\frac{7}{2880}+2R_4\Big]\varphi^3
	+\Big[\frac{1}{12}R_4+2R_6\Big]\varphi^5\\
&\qquad
	+R_4^2\varphi^7
	+\sum_{k=0}^3 O(n^k\varphi^{2k+5})
\Big).
\end{split}
\end{equation}
We can now rewrite the integrand of~(\ref{fct-J-zero-tau}). Setting
\[
y\sqrt{V}\varphi=\alpha\varphi=a
\]
we obtain, by combining~(\ref{lemme-sin2}) and~(\ref{lemme-rsin}),
\begin{equation}\label{da-rsinphi}
\begin{split}
\rho^n\,\frac{\sin(y\sqrt{V}\varphi-\chi)}{\sin\frac{\varphi}{2}}&=\\
	&e^{R_2\varphi^2}\Big(
	\frac{2}{\varphi}\sin(\alpha\varphi)+\sum_{k=1}^{11}J_k\varphi^k
	+\sum_{k=1}^8 O(n^k\varphi^{2k+4})\Big),
\end{split}
\end{equation}
where $J_{10}=0$ and
\begin{multicols}{2}
\begin{itemize}
\setlength{\itemsep}{1mm}
\item[$J_1=$\!] $\frac{1}{12}\sin(\alpha\varphi)$,
\item[$J_2=$\!] $-2\chi_3\cos(\alpha\varphi)$,
\item[$J_3=$\!] $\big[\frac{7}{2880}+2R_4\big]\sin(\alpha\varphi)$,
\item[$J_4=$\!] $-\big[\frac{1}{12}\chi_3+2\chi_5\big]\cos(\alpha\varphi)$,
\item[$J_5=$\!] $\big[\frac{1}{12}R_4+2R_6-\chi_3^2\big]
	\sin(\alpha\varphi)$,
\item[$J_6=$\!] $-2R_4\chi_3\cos(\alpha\varphi)$,
\item[$J_7=$\!] $\big[R_4^2-2\chi_3\chi_5-\frac{1}{24}\chi_3^2\big]
	\sin(\alpha\varphi)$,
\item[$J_8=$\!] $\frac{1}{3}\chi_3^3\cos(\alpha\varphi)$,
\item[$J_9=$\!] $-R_4\chi_3^2\sin(\alpha\varphi)$,
\item[$J_{11}=$\!] $\frac{1}{12}\chi_3^4\sin(\alpha\varphi)$.
\end{itemize}
\end{multicols}

(4) We will estimate~(\ref{fct-J-zero-tau}) using the form of the
integrand what we have obtained in~(\ref{da-rsinphi}).

We begin with the error terms. Applying Lemma~\ref{lemme-ebnphi}, we
have that
\begin{equation}\label{simpl-erreurs-int}
\Big|\int_0^{\tau}e^{R_2\varphi^2} f_k\,\mathrm{d}\varphi\Big|
\leq C n^k\int_0^{\infty}e^{-\frac{1}{2}V\varphi^2}\varphi^{2k+4}\,
\mathrm{d}\varphi=O\Big(\frac{1}{n^{5/2}}\Big),
\end{equation} 
with~$f_k=O(n^k\varphi^{2k+4})$, $k$ an integer between~1 and~8.

For the main part in~(\ref{da-rsinphi}) we can replace the integral on
$[0,\tau]$ by one on $[0,\infty)$ with an error of at most
$O(n^{-5/2})$. Indeed,
\begin{align*}
J^{**}(y)&:=
\Big|\int_{\tau}^{\infty}\!\!\!e^{R_2\varphi^2} \Big(
	\frac{2}{\varphi}\sin(\alpha\varphi)\!+\!\sum_{k=1}^{11} 
	J_k\varphi^k\Big)\mathrm{d}\varphi\Big|\\
&\leq \int_{\tau}^{\infty}\!\!\!e^{-\frac{1}{2}V\varphi^2} \Big(
	\frac{2}{\varphi}\!+\!\sum_{k=1}^{11} |J_k|\varphi^k\Big)
	\mathrm{d}\varphi.
\end{align*}
Using the definitions of the coefficients~$J_k$, and noting that
$\tau^{-1}=O(n^{1/4})$, we obtain with Lemma~\ref{lemme-ebnphi} that
\begin{equation}\label{simp-integ-tau-inf}
J^{**}(y)=O\Big(\frac{1}{n^{5/2}}\Big).
\end{equation}

Applying~(\ref{da-rsinphi}), (\ref{simpl-erreurs-int})
and~(\ref{simp-integ-tau-inf}) to~(\ref{fct-J-zero-tau}), we derive
that
\begin{equation*}
J(y)=\frac{1}{2\pi}\int_0^{\infty}e^{R_2\varphi^2}
	\frac{2}{\varphi}\sin(\alpha\varphi)\mathrm{d}\varphi
+\frac{1}{2\pi}\sum_{k=1}^{11}\int_0^{\infty}e^{R_2\varphi^2} 
	J_k\varphi^k\mathrm{d}\varphi
+O\Big(\frac{1}{n^{5/2}}\Big)
\end{equation*}
with $\alpha=y\sqrt{V}$ and~$R_2=-\frac{1}{2}V$. 

It remains to evaluate these integrals, which we have relegated to the
Appendix~\ref{annexe-fourier}. After simplification, using in
particular that $(p-q)^2=1-4pq$, we obtain that
\begin{equation}\label{simp-J-xi}
J(y)=\Phi(y)-\frac{1}{2}+\frac{e^{-\frac{1}{2}y^2}}{\sqrt{2\pi}}
	\Big(\frac{P_1(y)}{\sqrt{V}}+\frac{P_2(y)}{V}
	+\frac{P_3(y)}{V^{3/2}}+\frac{P_4(y)}{V^2}\Big)
	+O\Big(\frac{1}{n^{5/2}}\Big)
\end{equation}
with
\begin{itemize}[itemsep=1mm,leftmargin=1.8cm]
\item[$P_1(y)=$\!] $\frac{1}{6}(q-p)(1-y^2)$,
\item[$P_2(y)=$\!] $y[\frac{1}{72}(-3+7y^2-y^4)
	-\frac{pq}{36}(-3+11y^2-2y^4)]$,
\item[$P_3(y)=$\!] $(q-p)[
	\frac{1}{6480}(123+129y^2-384y^4+95y^6-5y^8)\\ 
	-\frac{pq}{3240}(3+69y^2-399y^4+145y^6-10y^8)]$,
\item[$P_4(y)=$\!] $y[
	\frac{1}{155520}(-4293-1359y^2+6165y^4-1971y^6+185y^8
		-5y^{10})\\
	+\frac{pq}{38880}(3105+1395y^2-7794y^4+2979y^6
		-325y^8+10y^{10})\\
	+\frac{p^2 q^2}{38880}(135-1035y^2+7947y^4-4167y^6
		+560y^8-20y^{10})]$.
\end{itemize}

(5) Finally, by Uspensky, we have that
\begin{equation}\label{usp}
\sum_{k=0}^j \binom{n}{k}p^k q^{n-k} = J(y)-J(y'),
\end{equation}
where 
\[
y=\frac{j-np+\frac{1}{2}}{\sqrt{V}}\quad\text{and}\quad
y'=-\frac{np+\frac{1}{2}}{\sqrt{V}}.
\]
Since $0<\liminf_{n\to\infty}p_n\leq \limsup_{n\to\infty}p_n<1$ there
are $\delta>0$ and $C>0$ such that, for large $n$,
\[
 \delta n\leq V\leq n \quad\text{and}\quad\delta \sqrt{n}\leq |y'| \leq C\sqrt{n}.
\]
It follows that for each $k=1,\dots,4$
\begin{equation}\label{pk}
\frac{e^{-\frac{1}{2}(y')^2}}{\sqrt{2\pi}}\frac{P_k(y')}{V^{k/2}}
=O\Big(\frac{1}{n^{5/2}}\Big).
\end{equation}
Moreover, for $x\leq -2$, $-\frac{1}{2}x^2\leq x$. Therefore, if $n$
is sufficiently large, then
\begin{equation}\label{phi}
\Phi(y') \leq\frac{1}{\sqrt{2\pi}}\int_{-\infty}^{y'}e^{x}\mathrm{d}x
=\frac{1}{\sqrt{2\pi}}e^{y'}=O\Big(\frac{1}{n^{5/2}}\Big),
\end{equation}
so that
\begin{equation}\label{jay}
J(y')=-\frac{1}{2}+O\Big(\frac{1}{n^{5/2}}\Big).
\end{equation}
Now the theorem follows from \eqref{usp}, \eqref{simp-J-xi} and
\eqref{jay}.
\end{proof}

We can deduce a first easy corollary.

\begin{corollaire}
  Suppose that $p=p_n\to p_0$ with $0<p_0<1$ and that $j=j_n$
  satisfies $\frac{j}{n}\to j_0$. Then
\[
\begin{aligned}
&\sum_{k=0}^j \binom{n}{k}p^k (1-p)^{n-k}=O\Big(\frac{1}{n^{5/2}}\Big)
	&\text{if }j_0<p_0,\\
&\sum_{k=0}^j \binom{n}{k}p^k (1-p)^{n-k}=1+O\Big(\frac{1}{n^{5/2}}\Big)
	&\text{if }j_0>p_0.
\end{aligned}
\]
\end{corollaire}

\begin{proof}
  From our assumptions it follows that
\begin{equation*}\label{da-y-cor}
\frac{y}{\sqrt{n}}=\frac{j-np+\frac{1}{2}}{n\sqrt{p(1-p)}} \to
\frac{j_0-p_0}{\sqrt{p_0(1-p_0)}}.
\end{equation*}
Then as in \eqref{pk} we have that, for $k=1,\ldots,4$,
\[
\frac{e^{-\frac{1}{2}y^2}}{\sqrt{2\pi}}\frac{P_k}{V^{k/2}} 
=O\Big(\frac{1}{n^{5/2}}\Big).
\]
Moreover, if $j_0<p_0$ then as in \eqref{phi} we have that
\[
\Phi(y)=O\Big(\frac{1}{n^{5/2}}\Big),
\]
while if $j_0>p_0$, we have that
\[
\Phi(y)= 1-\Phi(-y) = 1+ O\Big(\frac{1}{n^{5/2}}\Big).
\]
The claim now follows from Theorem~\ref{thm-general}.
\end{proof}

The second corollary will be crucial in our applications to lookback
options in the next section.

\begin{corollaire}\label{cor-approx-binom}
Suppose that 
\[
p=\frac{1}{2}+\frac{\alpha}{\sqrt{n}}+\frac{\beta}{n}
	+\frac{\gamma}{n^{3/2}}+\frac{\delta}{n^2}
	+\frac{\varepsilon}{n^{5/2}}
	+O\Big(\frac{1}{n^3}\Big)
\]
and 
\[
j=\frac{n}{2}+a\sqrt{n}+\frac{1}{2}+b_n+\frac{c}{\sqrt{n}}+\frac{d}{n}
	+\frac{e}{n^{3/2}}
	+O\Big(\frac{1}{n^2}\Big),
\]
where $(b_n)_n$ is a bounded sequence. Then
\[
\begin{split}
&\sum_{k=j}^n \binom{n}{k}p^k (1-p)^{n-k}
=\Phi(A)+\frac{e^{-\frac{1}{2}A^2}}{\sqrt{2\pi}}\Big(
	\frac{B_n}{\sqrt{n}}
	+\frac{C_0-C_2 B_n^2}{n}\\
&\quad +\frac{D_0-D_1 B_n-D_3 B_n^3}{n^{3/2}}
	+\frac{E_0-E_1 B_n-E_2 B_n^2+E_4 B_n^4}{n^2}\Big)
	+O\Big(\frac{1}{n^{5/2}}\Big),\\
\end{split}
\]
where
\begin{itemize}
\setlength{\itemsep}{1mm}
\item[$C_0=$\!] $2\alpha^2 A-(1-A^2)(A-8\alpha)/12+C$,
\item[$C_2=$\!] $A/2$,
\item[$D_0=$\!] $4\alpha\beta A+2(1-A^2)\beta/3+D$,
\item[$D_1=$\!] $(8\alpha A-1)/6-(1-A^2)
	(A^2-8\alpha A+24\alpha^2-3)/12+AC$,
\item[$D_3=$\!] $(1-A^2)/6$,
\item[$E_0=$\!] $2(\beta^2+2\alpha\gamma)A+(1-A^2)(6\alpha^2 C+2\gamma)/3
	+(3-A^2)(6\alpha^3-2C)\alpha A/3+(A^4-4A^2+1)(16\alpha^3-C)/12
	-(5A^6-53A^4+33A^2+171)A/1440+(5A^6-41A^4+21A^2+27)\alpha/90
	-(7A^4-40A^2+15)\alpha^2 A/18-AC^2/2+E$,
\item[$E_1=$\!] $4\beta A/3+(1-A^2)(2A-12\alpha)\beta/3+AD$,
\item[$E_2=$\!] $2\alpha^2 A+(1-A^2)(C+2\alpha^2 A)/2-(A^4-8A^2+9)A/24
	+(A^4-6A^2+3)\alpha/3$,
\item[$E_4=$\!] $(3-A^2)A/24$,
\end{itemize}
with
$A=2(\alpha-a)$, $B_n=2(\beta-b_n)$, $C=2(\gamma-c)$, $D=2(\delta-d)$,
$E=2(\varepsilon-e)$.
\end{corollaire}

\begin{proof}
  The proof is based on the same ideas as Lemma~3.1 of
  \cite{heuw1}. As it is very computational, we will just recall the
  main steps for obtaining the result.

We first note that by Theorem~\ref{thm-general},
\[
\begin{split}
  &\sum_{k=j}^n \binom{n}{k}p^k (1-p)^{n-k}= 1- \sum_{k=0}^{j-1}
\binom{n}{k}p^k (1-p)^{n-k}\\
  &\quad=1-\Phi(y)- \frac{e^{-\frac{1}{2}y^2}}{\sqrt{2\pi}}
  \Big(\frac{P_1}{\sqrt{V}}+\frac{P_2}{V}
  +\frac{P_3}{V^{3/2}}+\frac{P_4}{V^2}\Big)
  +O\Big(\frac{1}{n^{5/2}}\Big)
\end{split}
\]
with $y=\frac{j-np-1/2}{\sqrt{V}}$, where $P_1,\ldots,P_4$ and $V$ are
as in that theorem. The goal is then to give an asymptotic expansion
of each term in this sum.

We begin with $1-\Phi(y)=\Phi(-y)$.  This term is an integral that we
will decompose into two parts: one from~$-\infty$ to~$A$ (giving
$\Phi(A)$) and second one from~$A$ to~$-y$. This splitting is
motivated by the convergence of~$y$ to~$-A$ as~$n$ tends to
infinity. Then we find an asymptotic expansion for the second integral
by using a Taylor expansion about~$A$ in which we substitute~$y$ by
its asymptotic expansion.

We finish with all the other terms. For each of them, we
substitute~$V$ and~$y$ by their respective asymptotic expansions.
\end{proof}

\begin{rmq}\label{rmq-approx-binom}
We have stated the result for
\[
\binc{n,p}(j-1)=\sum_{k=j}^n\binom{n}{k}p^k (1-p)^{n-k},
\]
in line with the results in~\cite{chang_palmer}, \cite{lin_palmer}
and~\cite{heuw1}.  The corresponding result for the binomial
cumulative distribution function
\[
\bin{n,p}(j)=\sum_{k=0}^j\binom{n}{k}p^k (1-p)^{n-k}
\]
is obtained by writing~$j$ in the form
\[
j=\frac{n}{2}+a\sqrt{n}-\frac{1}{2}+b_n+\frac{c}{\sqrt{n}}+\frac{d}{n}
	+\frac{e}{n^{3/2}}+O\Big(\frac{1}{n^2}\Big)
\]
and taking 1 minus the result given in the corollary.
\end{rmq}

\section{Asymptotics of the price for lookback options}   

In this section we combine the results of the previous sections in
order to derive asymptotic expansions for the price of lookback
options. We will use the notation of Section 2. In addition we adopt
the following notations that are in line with common usage in the
literature:
\begin{itemize}
\setlength{\itemsep}{1mm}
\item[$d_1=$\!] $\frac{1}{\sigma\sqrt{\tau}}
	\big(\log\frac{S_t}{M_t}+(r+\frac{\sigma^2}{2})\tau\big)$,
\item[$d_2=$\!] $\frac{1}{\sigma\sqrt{\tau}}
	\big(\log\frac{S_t}{M_t}+(r-\frac{\sigma^2}{2})\tau\big)
	=d_1-\sigma\sqrt{\tau}$,
\item[$d_3=$\!] $\frac{1}{\sigma\sqrt{\tau}}
	\big(-\log\frac{S_t}{M_t}+(r-\frac{\sigma^2}{2})\tau\big)
	=-d_1+\frac{2r}{\sigma}\sqrt{\tau}$,
\item[$d_4=$\!] $\frac{1}{\sigma\sqrt{\tau}}
	\big(-\log\frac{S_t}{M_t}+(r+\frac{\sigma^2}{2})\tau\big)
	=d_3+\sigma\sqrt{\tau}$,
\end{itemize}
and we will write
\begin{itemize}
\setlength{\itemsep}{1mm}
\item[$\kappa_n=$\!] $\{j_0\}(1-\{j_0\})$,
\end{itemize}
where, as before, $j_0=\frac{\log (S_t/M_t)}{\sigma\sqrt{\tau/n}}$.

We first note that the price $C_{n}^{fl}$ of a lookback call as a
function of $n$ shows some (mild) oscillations, see Figure
\ref{fig-lookcall}, which are caused by the fact that the non-integer
part $\{j_0\}$ of the initial level $j_0$ varies with $n$.

\begin{figure}[ht!]
\begin{center}
\begin{tikzpicture}
\begin{axis}[legend pos=south east,xmin=0,xmax=400,ymax=26.4,x=0.1225mm]
\addplot[black,mark=none] coordinates{
(2,26.03214307
)(3,26.13641891
)(4,26.14768195
)(5,26.12857513
)(6,26.17119325
)(7,26.18015818
)(8,26.20223747
)(9,26.20591513
)(10,26.17827979
)(11,26.21961199
)(12,26.22453074
)(13,26.23289796
)(14,26.23556936
)(15,26.21968503
)(16,26.2365696
)(17,26.24022441
)(18,26.25244771
)(19,26.25494469
)(20,26.25413294
)(21,26.25558161
)(22,26.24426151
)(23,26.25836656
)(24,26.26077572
)(25,26.26873609
)(26,26.27051603
)(27,26.27157213
)(28,26.27276916
)(29,26.2679088
)(30,26.26857359
)(31,26.26971773
)(32,26.27783859
)(33,26.2793569
)(34,26.28322149
)(35,26.28437927
)(36,26.28444491
)(37,26.28526465
)(38,26.28195367
)(39,26.28245875
)(40,26.281928
)(41,26.28804165
)(42,26.28919755
)(43,26.29268602
)(44,26.29361429
)(45,26.29470039
)(46,26.29541227
)(47,26.29430589
)(48,26.29481304
)(49,26.2917101
)(50,26.29339471
)(51,26.29438361
)(52,26.29839353
)(53,26.29922873
)(54,26.30158541
)(55,26.30227279
)(56,26.30309117
)(57,26.30363688
)(58,26.30302634
)(59,26.30343667
)(60,26.30150069
)(61,26.30178195
)(62,26.30119121
)(63,26.3046507
)(64,26.30533664
)(65,26.30766209
)(66,26.30824586
)(67,26.30950214
)(68,26.30998722
)(69,26.31023672
)(70,26.31062667
)(71,26.30992912
)(72,26.31022755
)(73,26.30863994
)(74,26.30885045
)(75,26.30871013
)(76,26.31133929
)(77,26.31187657
)(78,26.31372056
)(79,26.31418599
)(80,26.31528272
)(81,26.31567842
)(82,26.31606456
)(83,26.31639265
)(84,26.31610359
)(85,26.31636624
)(86,26.31543601
)(87,26.31563536
)(88,26.31409664
)(89,26.31555344
)(90,26.31601431
)(91,26.31783984
)(92,26.31824799
)(93,26.31952202
)(94,26.31987875
)(95,26.32062418
)(96,26.32093081
)(97,26.32116988
)(98,26.32142776
)(99,26.32118203
)(100,26.32139249
)(101,26.32068282
)(102,26.32084719
)(103,26.31969372
)(104,26.32026871
)(105,26.32064854
)(106,26.32220401
)(107,26.3225447
)(108,26.32369187
)(109,26.32399425
)(110,26.32474774
)(111,26.32501268
)(112,26.32538674
)(113,26.3256151
)(114,26.32562362
)(115,26.32581626
)(116,26.32547273
)(117,26.32563052
)(118,26.32494807
)(119,26.32507186
)(120,26.3240632
)(121,26.32511485
)(122,26.32541966
)(123,26.32662864
)(124,26.32690403
)(125,26.32780846
)(126,26.32805499
)(127,26.3286644
)(128,26.32888264
)(129,26.32920637
)(130,26.32939689
)(131,26.32944405
)(132,26.32960743
)(133,26.32938691
)(134,26.32952373
)(135,26.32904421
)(136,26.32915504
)(137,26.32842496
)(138,26.32851037
)(139,26.32864912
)(140,26.32975181
)(141,26.32999136
)(142,26.33085822
)(143,26.33107538
)(144,26.33171288
)(145,26.33190804
)(146,26.33232255
)(147,26.33249609
)(148,26.33269383
)(149,26.33284615
)(150,26.33283322
)(151,26.33296472
)(152,26.33274709
)(153,26.33285815
)(154,26.33244166
)(155,26.33253267
)(156,26.33192302
)(157,26.33199439
)(158,26.33208499
)(159,26.33300962
)(160,26.33321174
)(161,26.33395295
)(162,26.33413756
)(163,26.33469984
)(164,26.33486721
)(165,26.33525495
)(166,26.33540535
)(167,26.33562286
)(168,26.33575657
)(169,26.3358081
)(170,26.33592539
)(171,26.3358151
)(172,26.33591623
)(173,26.33564819
)(174,26.33573346
)(175,26.33531167
)(176,26.33538133
)(177,26.33480971
)(178,26.33534541
)(179,26.33552608
)(180,26.33625256
)(181,26.33641923
)(182,26.3370022
)(183,26.33715505
)(184,26.33759764
)(185,26.33773688
)(186,26.33804217
)(187,26.33816798
)(188,26.33833902
)(189,26.3384516
)(190,26.33849137
)(191,26.33859091
)(192,26.33850235
)(193,26.33858905
)(194,26.33837503
)(195,26.33844909
)(196,26.33811246
)(197,26.33817406
)(198,26.33771759
)(199,26.33784999
)(200,26.33800711
)(201,26.33865699
)(202,26.33880284
)(203,26.33933736
)(204,26.33947209
)(205,26.33989351
)(206,26.34001724
)(207,26.3403278
)(208,26.34044069
)(209,26.34064258
)(210,26.34074476
)(211,26.34084015
)(212,26.34093176
)(213,26.34092278
)(214,26.34100396
)(215,26.34089271
)(216,26.34096361
)(217,26.34075216
)(218,26.34081292
)(219,26.34050329
)(220,26.34055405
)(221,26.34014826
)(222,26.34042592
)(223,26.34055976
)(224,26.34110506
)(225,26.34122977
)(226,26.341682
)(227,26.34179769
)(228,26.34215849
)(229,26.34226525
)(230,26.34253626
)(231,26.34263421
)(232,26.34281702
)(233,26.34290626
)(234,26.34300247
)(235,26.3430831
)(236,26.34309427
)(237,26.3431664
)(238,26.34309408
)(239,26.34315781
)(240,26.34300351
)(241,26.34305894
)(242,26.34282416
)(243,26.34287141
)(244,26.34255763
)(245,26.34259678
)(246,26.3424561
)(247,26.34296523
)(248,26.34307727
)(249,26.34350971
)(250,26.34361431
)(251,26.34397129
)(252,26.34406853
)(253,26.34435127
)(254,26.34444123
)(255,26.34465093
)(256,26.34473368
)(257,26.34487154
)(258,26.34494716
)(259,26.34501434
)(260,26.34508291
)(261,26.34508057
)(262,26.34514217
)(263,26.34507147
)(264,26.34512617
)(265,26.34498822
)(266,26.34503611
)(267,26.34483204
)(268,26.34487319
)(269,26.3446041
)(270,26.34463858
)(271,26.3444461
)(272,26.34489408
)(273,26.34499253
)(274,26.34537716
)(275,26.34546944
)(276,26.34579165
)(277,26.34587784
)(278,26.34613855
)(279,26.34621869
)(280,26.34641881
)(281,26.34649296
)(282,26.34663339
)(283,26.34670161
)(284,26.34678323
)(285,26.34684559
)(286,26.34686929
)(287,26.34692583
)(288,26.34689247
)(289,26.34694327
)(290,26.34685371
)(291,26.34689881
)(292,26.3467539
)(293,26.34679337
)(294,26.34659394
)(295,26.34662784
)(296,26.34637472
)(297,26.34640311
)(298,26.34646071
)(299,26.34683636
)(300,26.34692132
)(301,26.34724447
)(302,26.3473243
)(303,26.34759566
)(304,26.3476704
)(305,26.34789066
)(306,26.34796037
)(307,26.34813022
)(308,26.34819493
)(309,26.34831507
)(310,26.34837483
)(311,26.34844593
)(312,26.34850079
)(313,26.34852351
)(314,26.34857351
)(315,26.34854853
)(316,26.34859372
)(317,26.34852168
)(318,26.3485621
)(319,26.34844366
)(320,26.34847936
)(321,26.34831515
)(322,26.34834617
)(323,26.34813685
)(324,26.34816323
)(325,26.34801105
)(326,26.34835585
)(327,26.34843233
)(328,26.34873291
)(329,26.34880506
)(330,26.34906196
)(331,26.34912982
)(332,26.34934359
)(333,26.34940719
)(334,26.34957835
)(335,26.34963773
)(336,26.34976683
)(337,26.34982202
)(338,26.34990957
)(339,26.34996061
)(340,26.35000713
)(341,26.35005406
)(342,26.35006008
)(343,26.35010292
)(344,26.35006894
)(345,26.35010774
)(346,26.35003427
)(347,26.35006906
)(348,26.34995659
)(349,26.34998741
)(350,26.34983645
)(351,26.34986332
)(352,26.34967435
)(353,26.34969732
)(354,26.34960663
)(355,26.3499086
)(356,26.34997622
)(357,26.35024086
)(358,26.35030482
)(359,26.35053254
)(360,26.35059286
)(361,26.35078409
)(362,26.3508408
)(363,26.35099596
)(364,26.35104909
)(365,26.35116859
)(366,26.35121817
)(367,26.35130242
)(368,26.35134848
)(369,26.35139789
)(370,26.35144044
)(371,26.35145541
)(372,26.3514945
)(373,26.35147543
)(374,26.35151107
)(375,26.35145837
)(376,26.35149059
)(377,26.35140464
)(378,26.35143348
)(379,26.35131466
)(380,26.35134014
)(381,26.35118884
)(382,26.35121099
)(383,26.3510276
)(384,26.3511405
)(385,26.3512025
)(386,26.35145616
)(387,26.35151502
)(388,26.35173706
)(389,26.35179281
)(390,26.35198357
)(391,26.35203623
)(392,26.35219604
)(393,26.35224563
)(394,26.35237482
)(395,26.35242137
)(396,26.35252026
)(397,26.35256378
)(398,26.3526327
)(399,26.35267321
)(400,26.35271248)
};
\addplot[dashed] coordinates{(2,26.3864) (400,26.3864)};
\legend{$C_n^{fl}$,$C_{BS}^{fl}$}
\end{axis}
\end{tikzpicture}
\begin{tikzpicture}
\begin{axis}[xmin=50,xmax=150,ymin=26.29,ymax=26.34,x=0.49mm]
\addplot[black,mark=none] coordinates{
(2,26.03214307
)(3,26.13641891
)(4,26.14768195
)(5,26.12857513
)(6,26.17119325
)(7,26.18015818
)(8,26.20223747
)(9,26.20591513
)(10,26.17827979
)(11,26.21961199
)(12,26.22453074
)(13,26.23289796
)(14,26.23556936
)(15,26.21968503
)(16,26.2365696
)(17,26.24022441
)(18,26.25244771
)(19,26.25494469
)(20,26.25413294
)(21,26.25558161
)(22,26.24426151
)(23,26.25836656
)(24,26.26077572
)(25,26.26873609
)(26,26.27051603
)(27,26.27157213
)(28,26.27276916
)(29,26.2679088
)(30,26.26857359
)(31,26.26971773
)(32,26.27783859
)(33,26.2793569
)(34,26.28322149
)(35,26.28437927
)(36,26.28444491
)(37,26.28526465
)(38,26.28195367
)(39,26.28245875
)(40,26.281928
)(41,26.28804165
)(42,26.28919755
)(43,26.29268602
)(44,26.29361429
)(45,26.29470039
)(46,26.29541227
)(47,26.29430589
)(48,26.29481304
)(49,26.2917101
)(50,26.29339471
)(51,26.29438361
)(52,26.29839353
)(53,26.29922873
)(54,26.30158541
)(55,26.30227279
)(56,26.30309117
)(57,26.30363688
)(58,26.30302634
)(59,26.30343667
)(60,26.30150069
)(61,26.30178195
)(62,26.30119121
)(63,26.3046507
)(64,26.30533664
)(65,26.30766209
)(66,26.30824586
)(67,26.30950214
)(68,26.30998722
)(69,26.31023672
)(70,26.31062667
)(71,26.30992912
)(72,26.31022755
)(73,26.30863994
)(74,26.30885045
)(75,26.30871013
)(76,26.31133929
)(77,26.31187657
)(78,26.31372056
)(79,26.31418599
)(80,26.31528272
)(81,26.31567842
)(82,26.31606456
)(83,26.31639265
)(84,26.31610359
)(85,26.31636624
)(86,26.31543601
)(87,26.31563536
)(88,26.31409664
)(89,26.31555344
)(90,26.31601431
)(91,26.31783984
)(92,26.31824799
)(93,26.31952202
)(94,26.31987875
)(95,26.32062418
)(96,26.32093081
)(97,26.32116988
)(98,26.32142776
)(99,26.32118203
)(100,26.32139249
)(101,26.32068282
)(102,26.32084719
)(103,26.31969372
)(104,26.32026871
)(105,26.32064854
)(106,26.32220401
)(107,26.3225447
)(108,26.32369187
)(109,26.32399425
)(110,26.32474774
)(111,26.32501268
)(112,26.32538674
)(113,26.3256151
)(114,26.32562362
)(115,26.32581626
)(116,26.32547273
)(117,26.32563052
)(118,26.32494807
)(119,26.32507186
)(120,26.3240632
)(121,26.32511485
)(122,26.32541966
)(123,26.32662864
)(124,26.32690403
)(125,26.32780846
)(126,26.32805499
)(127,26.3286644
)(128,26.32888264
)(129,26.32920637
)(130,26.32939689
)(131,26.32944405
)(132,26.32960743
)(133,26.32938691
)(134,26.32952373
)(135,26.32904421
)(136,26.32915504
)(137,26.32842496
)(138,26.32851037
)(139,26.32864912
)(140,26.32975181
)(141,26.32999136
)(142,26.33085822
)(143,26.33107538
)(144,26.33171288
)(145,26.33190804
)(146,26.33232255
)(147,26.33249609
)(148,26.33269383
)(149,26.33284615
)(150,26.33283322
)(151,26.33296472
)(152,26.33274709
)(153,26.33285815
)(154,26.33244166
)(155,26.33253267
)(156,26.33192302
)(157,26.33199439
)(158,26.33208499
)(159,26.33300962
)(160,26.33321174
)(161,26.33395295
)(162,26.33413756
)(163,26.33469984
)(164,26.33486721
)(165,26.33525495
)(166,26.33540535
)(167,26.33562286
)(168,26.33575657
)(169,26.3358081
)(170,26.33592539
)(171,26.3358151
)(172,26.33591623
)(173,26.33564819
)(174,26.33573346
)(175,26.33531167
)(176,26.33538133
)(177,26.33480971
)(178,26.33534541
)(179,26.33552608
)(180,26.33625256
)(181,26.33641923
)(182,26.3370022
)(183,26.33715505
)(184,26.33759764
)(185,26.33773688
)(186,26.33804217
)(187,26.33816798
)(188,26.33833902
)(189,26.3384516
)(190,26.33849137
)(191,26.33859091
)(192,26.33850235
)(193,26.33858905
)(194,26.33837503
)(195,26.33844909
)(196,26.33811246
)(197,26.33817406
)(198,26.33771759
)(199,26.33784999
)(200,26.33800711
)(201,26.33865699
)(202,26.33880284
)(203,26.33933736
)(204,26.33947209
)(205,26.33989351
)(206,26.34001724
)(207,26.3403278
)(208,26.34044069
)(209,26.34064258
)(210,26.34074476
)(211,26.34084015
)(212,26.34093176
)(213,26.34092278
)(214,26.34100396
)(215,26.34089271
)(216,26.34096361
)(217,26.34075216
)(218,26.34081292
)(219,26.34050329
)(220,26.34055405
)(221,26.34014826
)(222,26.34042592
)(223,26.34055976
)(224,26.34110506
)(225,26.34122977
)(226,26.341682
)(227,26.34179769
)(228,26.34215849
)(229,26.34226525
)(230,26.34253626
)(231,26.34263421
)(232,26.34281702
)(233,26.34290626
)(234,26.34300247
)(235,26.3430831
)(236,26.34309427
)(237,26.3431664
)(238,26.34309408
)(239,26.34315781
)(240,26.34300351
)(241,26.34305894
)(242,26.34282416
)(243,26.34287141
)(244,26.34255763
)(245,26.34259678
)(246,26.3424561
)(247,26.34296523
)(248,26.34307727
)(249,26.34350971
)(250,26.34361431
)(251,26.34397129
)(252,26.34406853
)(253,26.34435127
)(254,26.34444123
)(255,26.34465093
)(256,26.34473368
)(257,26.34487154
)(258,26.34494716
)(259,26.34501434
)(260,26.34508291
)(261,26.34508057
)(262,26.34514217
)(263,26.34507147
)(264,26.34512617
)(265,26.34498822
)(266,26.34503611
)(267,26.34483204
)(268,26.34487319
)(269,26.3446041
)(270,26.34463858
)(271,26.3444461
)(272,26.34489408
)(273,26.34499253
)(274,26.34537716
)(275,26.34546944
)(276,26.34579165
)(277,26.34587784
)(278,26.34613855
)(279,26.34621869
)(280,26.34641881
)(281,26.34649296
)(282,26.34663339
)(283,26.34670161
)(284,26.34678323
)(285,26.34684559
)(286,26.34686929
)(287,26.34692583
)(288,26.34689247
)(289,26.34694327
)(290,26.34685371
)(291,26.34689881
)(292,26.3467539
)(293,26.34679337
)(294,26.34659394
)(295,26.34662784
)(296,26.34637472
)(297,26.34640311
)(298,26.34646071
)(299,26.34683636
)(300,26.34692132
)(301,26.34724447
)(302,26.3473243
)(303,26.34759566
)(304,26.3476704
)(305,26.34789066
)(306,26.34796037
)(307,26.34813022
)(308,26.34819493
)(309,26.34831507
)(310,26.34837483
)(311,26.34844593
)(312,26.34850079
)(313,26.34852351
)(314,26.34857351
)(315,26.34854853
)(316,26.34859372
)(317,26.34852168
)(318,26.3485621
)(319,26.34844366
)(320,26.34847936
)(321,26.34831515
)(322,26.34834617
)(323,26.34813685
)(324,26.34816323
)(325,26.34801105
)(326,26.34835585
)(327,26.34843233
)(328,26.34873291
)(329,26.34880506
)(330,26.34906196
)(331,26.34912982
)(332,26.34934359
)(333,26.34940719
)(334,26.34957835
)(335,26.34963773
)(336,26.34976683
)(337,26.34982202
)(338,26.34990957
)(339,26.34996061
)(340,26.35000713
)(341,26.35005406
)(342,26.35006008
)(343,26.35010292
)(344,26.35006894
)(345,26.35010774
)(346,26.35003427
)(347,26.35006906
)(348,26.34995659
)(349,26.34998741
)(350,26.34983645
)(351,26.34986332
)(352,26.34967435
)(353,26.34969732
)(354,26.34960663
)(355,26.3499086
)(356,26.34997622
)(357,26.35024086
)(358,26.35030482
)(359,26.35053254
)(360,26.35059286
)(361,26.35078409
)(362,26.3508408
)(363,26.35099596
)(364,26.35104909
)(365,26.35116859
)(366,26.35121817
)(367,26.35130242
)(368,26.35134848
)(369,26.35139789
)(370,26.35144044
)(371,26.35145541
)(372,26.3514945
)(373,26.35147543
)(374,26.35151107
)(375,26.35145837
)(376,26.35149059
)(377,26.35140464
)(378,26.35143348
)(379,26.35131466
)(380,26.35134014
)(381,26.35118884
)(382,26.35121099
)(383,26.3510276
)(384,26.3511405
)(385,26.3512025
)(386,26.35145616
)(387,26.35151502
)(388,26.35173706
)(389,26.35179281
)(390,26.35198357
)(391,26.35203623
)(392,26.35219604
)(393,26.35224563
)(394,26.35237482
)(395,26.35242137
)(396,26.35252026
)(397,26.35256378
)(398,26.3526327
)(399,26.35267321
)(400,26.35271248)
};
\end{axis}
\end{tikzpicture}
{\caption{Price of a lookback call, and its fine structure}\label{fig-lookcall}}
\end{center}
\end{figure}
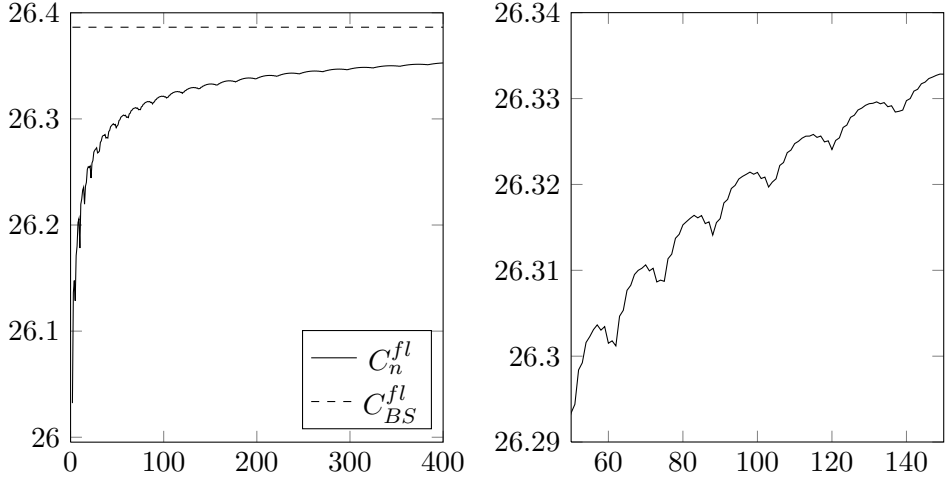

This reminds one of the (much more violent) oscillations for European
vanilla options. To deal with them, Diener and Diener
\cite{diener_diener} have used asymptotic expansions of the form
\begin{equation*}
C_n=c_0+\frac{c_1}{\sqrt{n}}+\frac{c_2}{n}
+O\Bigl(\frac{1}{n^{3/2}}\Bigr),
\end{equation*}
where the coefficients~$c_1$ et~$c_2$ are allowed to be bounded
functions of $n$. The variability of these coefficients capture the
observed oscillations.

In order to obtain such asymptotic expansions in our setting we will
need the following.

\begin{lemme}\label{l-frozen}
  Let $a_0$, $a_1$, $a_2\in \mathbb{R}$, and let $\eta=\eta(n)$ be a
  bounded function of~$n$. Then
\[
\Big(1+\frac{a_1}{\sqrt{n}}+\frac{a_2}{n}
+O\Bigl(\frac{1}{n^{3/2}}\Bigr)\Big)^{\eta} = 1 + \frac{a_1\eta}{\sqrt{n}} + \frac{a_2\eta-\frac{1}{2}\eta(1-\eta)a_1^2}{n} +O\Bigl(\frac{1}{n^{3/2}}\Bigr).
\]
\end{lemme}

\begin{proof} We have that, for large $n$,
\begin{align*}
\Big(1\!+\!\frac{a_1}{\sqrt{n}}\!+\!\frac{a_2}{n}\!
+\!O\Bigl(\frac{1}{n^{3/2}}\Bigr)\Big)^{\eta}
	&\!= \exp\Big(\eta \log\Big(1+\frac{a_1}{\sqrt{n}}+\frac{a_2}{n}
		+O\Bigl(\frac{1}{n^{3/2}}\Bigr)\Big)\Big)\\
	&\!= \exp\Big(\eta\Big(\frac{a_1}{\sqrt{n}}+\frac{a_2-\frac{1}{2}a_1^2}{n}
		+O\Bigl(\frac{1}{n^{3/2}}\Bigr)\Big)\Big)\\
	&\!= \exp\Big(\frac{a_1\eta}{\sqrt{n}}+\frac{(a_2-\frac{1}{2}a_1^2)\eta}{n}
		+O\Bigl(\frac{1}{n^{3/2}}\Bigr)\Big)\\
	&\!= 1+\frac{a_1\eta}{\sqrt{n}}+\frac{ (a_2-\frac{1}{2}a_1^2)\eta }{n}+
	\frac{1}{2} \frac{a_1^2\eta^2}{n}
		+O\Bigl(\frac{1}{n^{3/2}}\Bigr),
\end{align*}
which confirms the assertion. It is important to note here that each
big-O condition contains a constant that is absolute.
\end{proof}

In essence, the lemma says that we may expand the function
$(1+\frac{a_1}{\sqrt{n}}+\frac{a_2}{n} +O(\frac{1}{n^{3/2}}))^{\eta}$
as if $\eta$ was a constant. Diener and Diener \cite{diener_diener}
therefore speak of a \textit{frozen parameter}.

We first consider the asymptotic expansion of the lookback call. Its
price in the Black-Scholes model is well known. If $r>0$ then Goldman,
Sosin and Gatto \cite{gatto_goldman_sosin} found that
\begin{equation}\label{formula-ggs}
C^{fl}_{BS}(t)=S_t-S_t \theta_1 B_1-M_t B_2+S_t\,(1-\theta_2)B_3,
\end{equation}
where 
\begin{itemize}
\setlength{\itemsep}{1mm}
\item[$\theta_1=$\!]$1+\frac{\sigma^2}{2r}$,
\item[$\theta_2=$\!]$1-\frac{\sigma^2}{2r}$,
\item[$B_1=$\!]$\Phi(-d_1)$,
\item[$B_2=$\!]$e^{-r\tau}\Phi(d_2)$,
\item[$B_3=$\!]$e^{-r\tau}\big(\frac{S_t}{M_t}\big)^{-2r/\sigma^2}\Phi(d_3).$
\end{itemize}

By passing to the limit in (\ref{formula-ggs}), Babbs~\cite{babbs}
obtained the price in the case $r=0$ as
\begin{equation*}
C^{fl}_{BS}(t)= S_t-S_t B_1-M_t B_2-S_t(B_3^*-B_4^*),
\end{equation*}
where
\begin{itemize}
\item[$B_3^*=$\!]$\big(\log\frac{S_t}{M_t}+
	\frac{\sigma^2\tau}{2}\big)\Phi(-d_1)$,
\item[$B_4^*=$\!]$\sigma\sqrt{\tau}\,\frac{e^{-d_1^2/2}}{\sqrt{2\pi}}$.
\end{itemize}

\begin{theorem}\label{thm-call}
  Let $0\leq t< T$ and $n\in\mathbb{N}$. The price at time $t$ of the
  European lookback call option with floating strike in the $n$-period
  CRR binomial model satisfies the following:

\emph{(i)} if $r>0$ then   
\[
\begin{split}
C_{n}^{fl}(t)&=C^{fl}_{BS}(t)
	-S_t\frac{\sigma\sqrt{\tau}}{2}
		\big(\theta_1 B_1+\theta_2 B_3\big)
		\frac{1}{\sqrt{n}}\\
&-\Big[S_t\frac{\sigma^2 \tau}{12}
		\Big((\theta_1+2)B_1+(\theta_2+2-T_1)B_3\Big)
		-M_t T_2 B_4
		\Big]\frac{1}{n}\\
&+O\Bigl(\frac{1}{n^{3/2}}\Bigr),
\end{split}
\]
where $B_4=\sigma\sqrt{\tau}\,\big(\frac{S_t}{M_t}\big)^{(1-2r/\sigma^2)/2}\,
	\frac{e^{-\frac{1}{4}(d_1^2+d_4^2)}}{\sqrt{2\pi}}$,
$T_1=\frac{12r}{\sigma^2}\theta_2\kappa_n -(1+\frac{4r^2}{\sigma^4})\log\frac{S_t}{M_t}$ and $
T_2 =\frac{1}{2}+\kappa_n +\frac{d_4}{6\sigma\sqrt{\tau}}\log\frac{S_t}{M_t}$;
 
\emph{(ii)} if $r=0$, then 
\[
\begin{split}
C_{n}^{fl}(t)&=C^{fl}_{BS}(t)
	-S_t\frac{\sigma\sqrt{\tau}}{2}
		\big(2B_1+B_3^*-B_4^*\big)
		\frac{1}{\sqrt{n}}\\
&-\Big[S_t\frac{\sigma^2 \tau}{6}
		\Big(\big( 3+3\kappa_n-\frac{\sigma^2 \tau}{4}\big)B_1+B_3^*\Big)
		-S_t T_2^* B_4^*
		\Big]\frac{1}{n}\\
&+O\Bigl(\frac{1}{n^{3/2}}\Bigr),
\end{split}
\]
where $T_2^*=\frac{1}{2}+\kappa_n+\frac{\sigma^2 \tau}{12}
	-\frac{d_2}{6\sigma\sqrt{\tau}}\log\frac{S_t}{M_t}$.
\end{theorem}

\begin{proof} 
  The first step of the proof consists in writing~(\ref{somme-call})
  as a combination of (complementary) binomial cumulative distribution
  functions.

  As for $V_1$ we note that, by \eqref{pe} and \eqref{qu},
\[
u^{n-j_0-2k} q^k\,(1-q)^{n-k} = u^{-j_0} e^{-r\tau} p^k\,(1-p)^{n-k},
\]
so that
\[
V_1= \binc{n,q}(j_1-1)-\frac{M_t}{S_t}e^{-r\tau}\binc{n,p}(j_1-1),
\]
where $j_1=n-\lfloor \frac{n+j_0}{2}\rfloor$.

For $V_2$ we first make a change of index $k\to k+\joinf+1$ and then
proceed as for $V_1$ to obtain that
\[
V_2= Q^{-\joinf-1}\binc{n,q}(j_2-1)-\frac{M_t}{S_t}e^{-r\tau}P^{-\joinf-1}\binc{n,p}(j_2-1),
\]
where
\[
Q=\frac{q}{1-q},\quad P=\frac{p}{1-p}
\]
and $j_2=j_1+\joinf+1$.

For $V_3$ we proceed as in \cite{heuw1}. We first split the inner sum
and then change indices, $k\to k-j$ and $k\to k-j-1$. Next we
interchange the two double sums that have appeared, noting that $0\leq
j\leq n-\joinf-1$ and $0\leq k \leq\lfloor
\frac{n-j-\joinf-1}{2}\rfloor$ is equivalent to $0\leq k\leq \lfloor
\frac{n-\joinf-1}{2}\rfloor$ and $0\leq j\leq n-\joinf-1-2k$; and that
$0\leq j\leq n-\joinf-1$ and $0\leq k \leq\lfloor
\frac{n-j-\joinf-1}{2}\rfloor-1$ is equivalent to $0\leq k\leq \lfloor
\frac{n-\joinf-1}{2}\rfloor-1$ and $0\leq j\leq n-\joinf-3-2k$. Note
that $\lfloor \frac{n-\joinf-1}{2}\rfloor=j_1-1$.

Now, if $r>0$, then the geometric series that arise as inner sums have
ratio different from 1, so that by continuing as in \cite{heuw1} we
obtain
\begin{align*}
V_3 =& \frac{Q(1-d)}{(Q-1)(Qd-1)} \,\big(\bin{n,q}(j_3) - Q\bin{n,q}(j_3-1)\big)\\
&+ \frac{Q^{-\joinf-1}}{Q-1}\,\big(Q\bin{n,1-q}(j_3) -\bin{n,1-q}(j_3-1)\big)\\
&+ e^{-r\tau}\frac{(Qd)^{-\joinf-1}}{d(1-Qd)}\,\big(P\bin{n,1-p}(j_3) -\bin{n,1-p}(j_3-1)\big)
\end{align*}
with $j_3=j_1-1$, where we have used that $u^{-1}=d$.

However, if $r=0$, then $Q=u$ so that geometric series with ratio 1
appear. Using $k\binom{n}{k} = n \binom{n-1}{k-1}$, the calculations
then lead us to
\begin{align*}
V_3 =& \Big(\joinf-n-\frac{1}{u-1}\Big) \,\big(\bin{n,q}(j_3) - u\bin{n,q}(j_3-1)\big) - 2u\bin{n,q}(j_3-1)\\
&+ \frac{u^{-\joinf-1}}{u-1}\,\big(u\bin{n,p}(j_3) -\bin{n,p}(j_3-1)\big)\\
&+ 2nq\big(\bin{n-1,q}(j_3-1) -u\bin{n-1,q}(j_3-2)\big),
\end{align*}
where we have used that $1-q=p$ in this situation. We note that the
formulas for $V_1$ and $V_2$ are not affected, but one may replace $Q$
by $u$ and $P$ by $d$ in $V_2$.

The second step of the proof is to expand each term using
Corollary~\ref{cor-approx-binom} and Remark \ref{rmq-approx-binom}. We
note that we use Lemma \ref{l-frozen} with $\eta=\{j_0\}$ to expand a
power containing $\joinf$. For example, we write
\[
Q^{-\joinf-1} = Q^{\{j_0\}}Q^{-j_0-1}
\]
and expand separately. We then obtain the result after a long series
of calculations and simplifications.
\end{proof}

We turn to the asymptotics of the put price. We have to replace the
running minimum by the running maximum
\[
M_t:=\max_{t^*\leq t}S_{t^*},
\]
and the initial level becomes
\[
j_0=\frac{\log (M_t/S_t)}{\sigma\sqrt{\tau/n}} \; (>0).
\]
With this re-interpretation, the numbers $d_1$, $d_2$, $d_3$, $d_4$
and $\kappa_n$ keep their meaning.

For $r>0$, Goldman, Sosin and Gatto \cite{gatto_goldman_sosin} found that 
\begin{equation*}\label{formula-put-ggs}
P^{fl}_{BS}(t)=-S_t+S_t \theta_1 B_1+M_t B_2-S_t\,(1-\theta_2)B_3
\end{equation*}
with
\begin{itemize}
\setlength{\itemsep}{1mm}
\item[$B_1=$\!]$\Phi(d_1)$,
\item[$B_2=$\!]$e^{-r\tau}\Phi(-d_2)$,
\item[$B_3=$\!]$e^{-r\tau}\big(\frac{S_t}{M_t}\big)^{-2r/\sigma^2}\Phi(-d_3)$,
\end{itemize}
and, for $r=0$, Babbs \cite{babbs} obtained that
\begin{equation*}
P^{fl}_{BS}(t)= -S_t+S_t B_1+M_t B_2+S_t (B_3^*+B_4^*)
\end{equation*}
with
\begin{itemize}
\setlength{\itemsep}{1mm}
\item[$B_3^*=$\!]$\big(\log\frac{S_t}{M_t}+
	\frac{\sigma^2\tau}{2}\big)\Phi(d_1)$,
\item[$B_4^*=$\!]$\sigma\sqrt{\tau}\,\frac{e^{-d_1^2/2}}{\sqrt{2\pi}}$.
\end{itemize}

\begin{theorem}\label{thm-put}
  Let $0\leq t< T$ and $n\in\mathbb{N}$. The price at time $t$ of the
  European lookback put option with floating strike in the $n$-period
  CRR binomial model satisfies the following:

\emph{(i)} if $r>0$ then   
\[
\begin{split}
P_{n}^{fl}(t)&=P^{fl}_{BS}(t)
	-S_t\frac{\sigma\sqrt{\tau}}{2}
		\big(\theta_1 B_1+\theta_2 B_3\big)
		\frac{1}{\sqrt{n}}\\
&+\Big[S_t\frac{\sigma^2 \tau}{12}
		\Big((\theta_1+2)B_1+(\theta_2+2-T_1)B_3\Big)
		+M_t T_2 B_4
		\Big]\frac{1}{n}\\
&+O\Bigl(\frac{1}{n^{3/2}}\Bigr),
\end{split}
\]
where $B_4=\sigma\sqrt{\tau}\,\big(\frac{S_t}{M_t}\big)^{(1-2r/\sigma^2)/2}\,
	\frac{e^{-\frac{1}{4}(d_1^2+d_4^2})}{\sqrt{2\pi}}$,
$T_1=\frac{12r}{\sigma^2}\theta_2\kappa_n
	-(1+\frac{4r^2}{\sigma^4})\log\frac{S_t}{M_t}$
and
$T_2=\frac{1}{2}+\kappa_n
	+\frac{d_4}{6\sigma\sqrt{\tau}}\log\frac{S_t}{M_t}$;
	
\emph{(ii)} if $r=0$ then 
\[
\begin{split}
P_{n}^{fl}(t)&=P^{fl}_{BS}(t)
	-S_t\frac{\sigma\sqrt{\tau}}{2}
		\big(2B_1+B_3^*+B_4^*\big)
		\frac{1}{\sqrt{n}}\\
&+\Big[S_t\frac{\sigma^2 \tau}{6}
		\Big(\big( 3+3\kappa_n-\frac{\sigma^2 \tau}{4}\big)B_1+B_3^*\Big)
		+S_t T_2^* B_4^*
		\Big]\frac{1}{n}\\
&+O\Bigl(\frac{1}{n^{3/2}}\Bigr),
\end{split}
\]
where $T_2^*=\frac{1}{2}+\kappa_n+\frac{\sigma^2 \tau}{12}
	-\frac{d_2}{6\sigma\sqrt{\tau}}\log\frac{S_t}{M_t}$.
\end{theorem}

The proof is quite similar to the proof of Theorem~\ref{thm-call} and
is therefore omitted.

\begin{rmq}
  In case the option is valued at emission we have that~$S_t=M_t$ and
  thus~$\kappa_n=0$.  The formulas obtained here then reduce to the
  formulas previously found by the second author~\cite{heuw1}.
\end{rmq}

\begin{rmq}\label{rmq-terme2}
  In Theorems~\ref{thm-call} and~\ref{thm-put}, we have that~the
  coefficients of $\frac{1}{n}$ are bounded functions of~$n$. In fact,
  these coefficients are affine functions
  of~$\kappa_n=\{j_0\}(1-\{j_0\})$, which is bounded in $n$. The
  function $x\mapsto x(1-x)$ is therefore responsible for the
  parabola-like oscillations in Figure \ref{fig-lookcall}.
\end{rmq}

\begin{rmq}
The coefficients for~$r=0$ are the limits of those for~$r> 0$.
\end{rmq}

%%%%%%%%%%%%%%%%%%%%%%%%%%%%%%%%%%
%%%%%%%%%%%%%%%%%%%%%%%%%%%%%%%%%%
%%%%%%%%%%%%%%%%%%%%%%%%%%%%%%%%%%

\section{Numerical examples}
In this part we give a numerical illustration of
Theorems~\ref{thm-call} and Theorem~\ref{thm-put}. These results tell
us that
\[
C_{n}^{fl}=C^{fl}_{BS} +\frac{C_1}{\sqrt{n}}+ \frac{C_2}{n}+ O\Bigl(\frac{1}{n^{3/2}}\Bigr)
\]
and
\[
P_{n}^{fl}=P^{fl}_{BS} +\frac{P_1}{\sqrt{n}}+ \frac{P_2}{n}+ O\Bigl(\frac{1}{n^{3/2}}\Bigr)
\]
with certain constants $C_1$, $P_1$ and functions $C_2$, $P_2$ that
are bounded in $n$. For example, for the call we should find that
$(C^{fl}_n-C^{fl}_{BS})\sqrt{n}$ and
$(C^{fl}_n-C^{fl}_{BS}-C_1/\sqrt{n})n$ almost coincide respectively
with $C_1$ and~$C_2$ for large~$n$.

This will be considered in the four tables below. We choose as values
$S_0=80$, $\sigma=0.2$ and $\tau=1.27$. For each type of option we
produce an example with a positive value for the spot rate~$(r=0.08)$
and an example with $r=0$.

For the call, we take~$M_t=60$ as the minimal price of the underlying
(see Tables~\ref{tab-call-rpos} and~\ref{tab-call-rnul}). The results
obtained are consistent with Theorem~\ref{thm-call}.

\begin{table}[h!]
\caption{Example for the call ($r>0$)}\label{tab-call-rpos}
\begin{tabular}{lrrrrr}
\hline
Number of periods~$n$ & 1,000 & 5,000 & 10,000& 50,000&100,000\tabularnewline
\hline
$C^{fl}_n $           &26.3647&26.3765&26.3794&26.3832&26.3842\tabularnewline
$C^{fl}_{BS}$         &26.3864&26.3864&26.3864&26.3864&26.3864\tabularnewline
$(C^{fl}_n-C^{fl}_{BS})\sqrt{n}$                    
                      &-0.6866&-0.6987&-0.7004&-0.7040&-0.7050\tabularnewline
$C_1$                 &-0.7071&-0.7071&-0.7071&-0.7071&-0.7071\tabularnewline
$(C^{fl}_n-C^{fl}_{BS}-C_1/\sqrt{n})n$
                      &0.6491&0.5931&0.6658&0.6868&0.6746\tabularnewline
$C_2$                 &0.6640&0.5961&0.6635&0.6808&0.6681\tabularnewline
\hline
\end{tabular}
\end{table}
\begin{table}[ht!]
\caption{Example for the call ($r=0$)}\label{tab-call-rnul}
\begin{tabular}{lrrrrr}
\hline
Number of periods~$n$ & 1,000 & 5,000 & 10,000& 50,000&100,000\tabularnewline
\hline
$C^{fl}_n $           &21.3779&21.4016&21.4074&21.4151&21.4169\tabularnewline
$C^{fl}_{BS}$         &21.4214&21.4214&21.4214&21.4214&21.4214\tabularnewline
$(C^{fl}_n-C^{fl}_{BS})\sqrt{n}$                    
                      &-1.3755&-1.3956&-1.3985&-1.4044&-1.4060\tabularnewline
$C_1$                 &-1.4095&-1.4095&-1.4095&-1.4095&-1.4095\tabularnewline
$(C^{fl}_n-C^{fl}_{BS}-C_1/\sqrt{n})n$
                      &1.0746&0.9868&1.1024&1.1371&1.1173\tabularnewline
$C_2$                 &1.1144&1.0069&1.1136&1.1410&1.1209\tabularnewline
\hline
\end{tabular}
\end{table}

We take the same parameters for the put but here we consider the
maximal price of the underlying at time $t$ as $M_t=100$ (see
Tables~\ref{tab-put-rpos} and~\ref{tab-put-rnul}). The results are
consistent with Theorem~\ref{thm-put}.

\begin{table}[h!]
\caption{Example for the put ($r>0$)}\label{tab-put-rpos}
\begin{tabular}{lrrrrr}
\hline
Number of periods~$n$ & 1,000 & 5,000 & 10,000& 50,000&100,000\tabularnewline
\hline
$P^{fl}_n $           &16.3662&16.4536&16.4747&16.5031&16.5098\tabularnewline
$P^{fl}_{BS}$         &16.5260&16.5260&16.5260&16.5260&16.5260\tabularnewline
$(P^{fl}_n-P^{fl}_{BS})\sqrt{n}$                    
                      &-5.0523&-5.1200&-5.1274&-5.1325&-5.1394\tabularnewline
$P_1$                 &-5.1466&-5.1466&-5.1466&-5.1466&-5.1466\tabularnewline
$(P^{fl}_n-P^{fl}_{BS}-P_1/\sqrt{n})n$
                      &2.9814&1.8813&1.9153&3.1439&2.2781\tabularnewline
$P_2$                 &3.0671&1.9652&1.9524&3.1866&2.3146\tabularnewline
\hline
\end{tabular}
\end{table}

\begin{table}[h!]
\caption{Example for the put ($r=0$)}\label{tab-put-rnul}
\begin{tabular}{lrrrrr}
\hline
Number of periods~$n$ & 1,000 & 5,000 & 10,000& 50,000&100,000\tabularnewline
\hline
$P^{fl}_n $           &23.4800&23.5410&23.5559&23.5759&23.5806\tabularnewline
$P^{fl}_{BS}$         &23.5921&23.5921&23.5921&23.5921&23.5921\tabularnewline
$(P^{fl}_n-P^{fl}_{BS})\sqrt{n}$                    
                      &-3.5462&-3.6140&-3.6217&-3.6271&-3.6340\tabularnewline
$P_1$                 &-3.6413&-3.6413&-3.6413&-3.6413&-3.6413\tabularnewline
$(P^{fl}_n-P^{fl}_{BS}-P_1/\sqrt{n})n$
                      &3.0079&1.9330&1.9554&3.1721&2.3143\tabularnewline
$P_2$                 &3.0623&1.9703&1.9577&3.1807&2.3166\tabularnewline
\hline
\end{tabular}
\end{table}

%%%%%%%%%%%%%%%%%%%%%%
%%%%%%%%%%%%%%%%%%%%%%
%%%%%%%%%%%%%%%%%%%%%%

\section{Conclusion}
The main goal of our paper is to derive an asymptotic expansion in
powers of $n^{-1/2}$ for the price of European lookback options with
floating strike. In order to achieve this we had to refine a discrete
model of Cheuk and Vorst. Their tree only worked for options evaluated
at emission. In our work we consider the price of the option at any
time between emission and maturity. This more general situation turned
out to require a new type of tree that mixes a partial binomial tree
with a Cheuk-Vorst tree. Counting the number of paths in this tree, we
derive a closed formula for the option price.

In order to describe the asymptotic behaviour of this formula we need
an asymptotic expansion of the binomial cumulative distribution
function with a smaller error term than known so far in the
literature. Following the work of Chang, Lin and
Palmer~\cite{chang_palmer}, \cite{lin_palmer}, we base our work on an
integral representation of the binomial cumulative distribution
function that is due to Uspensky~\cite{uspensky}.

We procure explicit formulas for the coefficients of $n^{-1/2}$ and
$n^{-1}$ in the asymptotic expansion, both for the call and the put,
and for any values of the parameters; in particular, we allow the spot
rate to be zero. These formulas confirm the convergence to the Black
Scholes prices. Our results are tested on random examples.

Several issues can be proposed as a follow-up to our work. One can
continue the study on lookback options; so far no asymptotic
expansions are known for the fixed-strike case. Cheuk and
Vorst~\cite{cheuk_vorst} built a one-state tree which can be the basis
for this work. However, the price deduced from their tree cannot, a
priori, be written using a binomial cumulative distribution
function. It will be necessary to provide an asymptotic expansion for
some new type of functions.

Another possibility is to look at Asian options for which the payoff
is determined by the average value of the underlying. The average can
be taken to be geometric or arithmetic. Again there are fixed-strike
and floating strike options. So far, no closed form for the price is
known in the case of arithmetic
averages~\cite{conze_viswanathan}. Thus determining the asymptotic
expansion of the price obtained by an equivalent CRR tree could allow
to find a closed-form solution also in the arithmetic case.
%%%%%%%%%%%%%%%%%%%%%%
%%%%%%%%%%%%%%%%%%%%%%
%%%%%%%%%%%%%%%%%%%%%%

%%%%%%%%%%%%%%%%%%%%%%
%%%%%%%%%%%%%%%%%%%%%%
%%%%%%%%%%%%%%%%%%%%%%

\appendix

\section{Some integrals}\label{annexe-fourier}

We evaluate here some integrals that are needed in the proof of
Theorem \ref{thm-general}. Recall that $R_2=-\frac{1}{2}V$ and
$\alpha= y\sqrt{V}$. The first integral is a Fourier sine transform:
\begin{equation*}
\begin{split}
\frac{1}{2\pi}\int_0^{\infty}
	e^{R_2\varphi^2}\frac{2}{\varphi}\sin(\alpha\varphi)\,\mathrm{d}\varphi
&= \frac{1}{\pi}\int_0^{\infty}
	\frac{e^{-\frac{1}{2}x^2}}{x}\sin(y x)\,\mathrm{d}x\\
&= \frac{1}{\sqrt{2\pi}}\int_0^{y} e^{-\frac{1}{2}x^2}\,\mathrm{d}x
	= \Phi(y)-\frac{1}{2}
\end{split}
\end{equation*}
with $\Phi$ the standard normal cumulative distribution function; see
\cite[pp. 128--129]{uspensky}, \cite[p. 73]{erdelyi}.

The remaining integrals are Fourier sine or Fourier cosine transforms
of the functions $x\mapsto x^m e^{-\frac{1}{2}x^2}$, $m\geq 1$. It is
well known that their values involve the Hermite polynomials $H_m$:
\[
\frac{1}{\pi}\int_0^{\infty} x^me^{-\frac{1}{2}x^2}
\left.\begin{cases}
\sin(y x)&\text{if $m$ is odd}\\
\cos(y x)&\text{if $m$ is even}
\end{cases}\right\}
\,\mathrm{d}x
= (-1)^{\lfloor m/2\rfloor} \frac{e^{-\frac{1}{2}y^2}}{\sqrt{2\pi}}\,H_m(y),
\]
see \cite[p. 15, p. 74]{erdelyi}. For $m=1$ we thus have
\begin{equation*}
\begin{split}
\frac{1}{2\pi}\int_0^{\infty}
	e^{R_2\varphi^2}J_1\varphi\,\mathrm{d}\varphi
&= \frac{1}{24V}\,\frac{1}{\pi}\int_0^{\infty}
	xe^{-\frac{1}{2}x^2}\sin(y x)\,\mathrm{d}x\\
&= \frac{1}{24V}\,\frac{e^{-\frac{1}{2}y^2}}{\sqrt{2\pi}}\,y.
\end{split}
\end{equation*}
For $m=2$,
\begin{equation*}
\begin{split}
\frac{1}{2\pi}\int_0^{\infty}
	e^{R_2\varphi^2}J_2\varphi^2\,\mathrm{d}\varphi
&= -\frac{p-q}{6\sqrt{V}}\,\frac{1}{\pi}\int_0^{\infty}
	x^2e^{-\frac{1}{2}x^2}\cos(y x)\,\mathrm{d}x\\
	&=\frac{p-q}{6\sqrt{V}}\,\frac{e^{-\frac{1}{2}y^2}}{\sqrt{2\pi}}\,(y^2-1).
\end{split}
\end{equation*}
For the following values of $m$ we obtain
\begin{equation*}
\begin{split}
\frac{1}{2\pi}\int_0^{\infty}
	e^{R_2\varphi^2}J_3\varphi^3\,\mathrm{d}\varphi
&= \frac{J_3^*}{2V^2}\,\frac{1}{\pi}\int_0^{\infty}
	x^3e^{-\frac{1}{2}x^2}\sin(y x)\,\mathrm{d}x\\
&=\frac{J_3^*}{2V^2}\,
	\frac{e^{-\frac{1}{2}y^2}}{\sqrt{2\pi}}\,y(3-y^2)
\end{split}
\end{equation*}
with $J_3^*=\frac{7}{2880}
	+\frac{1}{2}V\big(\frac{1}{6}-pq\big)$; 	
\begin{equation*}
\begin{split}
\frac{1}{2\pi}\int_0^{\infty}
	e^{R_2\varphi^2}J_4\varphi^4\,\mathrm{d}\varphi
&= \frac{J_4^*(p-q)}{2 V^{3/2}}\,\frac{1}{\pi}\int_0^{\infty}
	x^4e^{-\frac{1}{2}x^2}\cos(y x)\,\mathrm{d}x\\
	&=\frac{J_4^*(p-q)}{2V^{3/2}}\,
	\frac{e^{-\frac{1}{2}y^2}}{\sqrt{2\pi}}\,(3-6y^2+y^4)
\end{split}
\end{equation*}
with $J_4^*=-\frac{1}{72}
	+\frac{1}{60}\big(1-12pq\big)$;
\begin{equation*}
\begin{split}
\frac{1}{2\pi}\int_0^{\infty}
	e^{R_2\varphi^2}J_5\varphi^5\,\mathrm{d}\varphi
&= \frac{J_5^*}{2 V^2}\,\frac{1}{\pi}\int_0^{\infty}
	x^5e^{-\frac{1}{2}x^2}\sin(y x)\,\mathrm{d}x\\
&= \frac{J_5^*}{2V^2}
	\frac{e^{-\frac{1}{2}y^2}}{\sqrt{2\pi}}\,y(15-10y^2+y^4)
\end{split}
\end{equation*}
with $J_5^*=\frac{1}{48}\big(\frac{1}{6}-pq\big)
	-\frac{1}{3}\big(\frac{1}{120}-\frac{1}{4}pq+p^2 q^2\big)
	-\frac{1}{36}V(p-q)^2$;
\begin{equation*}
\begin{split}
\frac{1}{2\pi}\int_0^{\infty}
	e^{R_2\varphi^2}J_6\varphi^6\,\mathrm{d}\varphi
&= \frac{J_6^*(p-q)}{2V^{3/2}}\,\frac{1}{\pi}\int_0^{\infty}
	x^6e^{-\frac{1}{2}x^2}\cos(y x)\,\mathrm{d}x\\
&= \frac{J_6^*(p-q)}{2V^{3/2}}
	\frac{e^{-\frac{1}{2}y^2}}{\sqrt{2\pi}}\,(15-45y^2+15y^4-y^6)
\end{split}
\end{equation*}
with $J_6^*=-\frac{1}{12}(\frac{1}{6}-pq)$;
\begin{equation*}
\begin{split}
\frac{1}{2\pi}\int_0^{\infty}
	e^{R_2\varphi^2}J_7\varphi^7\,\mathrm{d}\varphi
&= \frac{J_7^*}{2V^2}\,\frac{1}{\pi}\int_0^{\infty}
	x^7e^{-\frac{1}{2}x^2}\sin(y x)\,\mathrm{d}x\\
&= \frac{J_7^*}{2V^2}
	\frac{e^{-\frac{1}{2}y^2}}{\sqrt{2\pi}}\,y(105-105y^2+21y^4-y^6)
\end{split}
\end{equation*}
with $J_7^*=\frac{1}{16}(\frac{1}{6}-pq)^2+\frac{1}{360}(p-q)^2(1-12pq)
	-\frac{1}{864}(p-q)^2$;
\begin{equation*}
\begin{split}
\frac{1}{2\pi}\int_0^{\infty}
	e^{R_2\varphi^2}J_8\varphi^8\,\mathrm{d}\varphi
&= \frac{(p-q)^3}{1296V^{3/2}}\,\frac{1}{\pi}\int_0^{\infty}
	x^8e^{-\frac{1}{2}x^2}\cos(y x)\,\mathrm{d}x\\
&= \frac{(p-q)^3}{1296 V^{3/2}}\,
	\frac{e^{-\frac{1}{2}y^2}}{\sqrt{2\pi}}\,H_8(y)
\end{split}
\end{equation*}
with $H_8(y)=105\!-\!420y^2\!+\!210y^4\!-\!28y^6\!+\!y^8$;
\begin{equation*}
\begin{split}
\frac{1}{2\pi}\int_0^{\infty}
	e^{R_2\varphi^2}J_9\varphi^9\,\mathrm{d}\varphi
&= \frac{J_9^*(p-q)^2}{2 V^2}\,\frac{1}{\pi}\int_0^{\infty}
	x^9e^{-\frac{1}{2}x^2}\sin(y x)\,\mathrm{d}x\\
&= \frac{J_9^*(p-q)^2}{2 V^2}\,
	\frac{e^{-\frac{1}{2}y^2}}{\sqrt{2\pi}}\,H_9(y)
\end{split}
\end{equation*}
with $J_9^*=-\frac{1}{144}(\frac{1}{6}-pq)$ and
$H_9(y)=y(945\!-\!1260y^2\!+\!378y^4\!-\!36y^6\!+\!y^8)$;
\begin{equation*}
\begin{split}
\frac{1}{2\pi}\int_0^{\infty}
	e^{R_2\varphi^2}J_{11}\varphi^{11}\,\mathrm{d}\varphi
&= \frac{(p-q)^4}{31104V^2}\,\frac{1}{\pi}\int_0^{\infty}
	x^{11}e^{-\frac{1}{2}x^2}\sin(y x)\,\mathrm{d}x\\
&= -\frac{(p-q)^4}{31104 V^2}\,
	\frac{e^{-\frac{1}{2}y^2}}{\sqrt{2\pi}}\,H_{11}(y)
\end{split}
\end{equation*}
with $H_{11}(y)=y(-10395\!+\!17325y^2\!-\!6930y^4\!
+\!990y^6\!-\!55y^8\!+\!y^{10})$.

\end{document}